%% file: main.tex
\documentclass{lmcs}

\usepackage{etoolbox}
\makeatletter
\patchcmd{\@tocline}
  {\hfil\hbox to\@pnumwidth{\@tocpagenum{#7}}\par}
  {\dotfill\hbox to\@pnumwidth{\@tocpagenum{#7}}\par}
  {}{}
\makeatother
\usepackage{subcaption}
\usepackage{calc}
\usepackage[shortlabels]{enumitem}
\setlist{nosep}
\usepackage{booktabs}
\usepackage{longtable}
\usepackage[width=.8\textwidth]{caption}
\usepackage{mathtools}
\allowdisplaybreaks[1]
\counterwithout{equation}{section}
\usepackage{nicefrac}
\usepackage{bm}
\usepackage{proof}
\usepackage{savesym}
\savesymbol{openbox}
\usepackage{amsthm}

\usepackage{thmtools}
\savesymbol{Bbbk}
\usepackage{amssymb}
\usepackage{stmaryrd}
\usepackage{array}
\usepackage{tabularx}
\newcolumntype{Y}{>{\centering\arraybackslash}X}
\newcolumntype{C}[1]{>{\centering\arraybackslash}p{#1}}
\usepackage{multirow}
\usepackage[dvipsnames]{xcolor}

\usepackage[most]{tcolorbox}
\newtcolorbox{figurebox}{
  colback=white,
  enhanced,
  colframe=black,
  boxrule=1pt,
  arc=0pt,
  width=\textwidth
}
\usepackage[
  backend=biber,
  style=numeric,
  sorting=nyt,
  sortcites=true,
  backref=true,
  useprefix=true,
  maxnames=8,
  minnames=7,
  maxalphanames=4,
  minalphanames=3
]{biblatex}
\newbibmacro*{printdoi}{\iffieldundef{doi}{}{\printtext{[DOI: }\href{https://doi.org/\thefield{doi}}{\nolinkurl{\thefield{doi}}}\printtext{]}}}
\newbibmacro*{printeprint}{\iffieldundef{eprint}{}{\printtext{[PDF: }\href{https://arxiv.org/abs/\thefield{eprint}}{\nolinkurl{arXiv:\thefield{eprint}}}\printtext{]}}}
\newbibmacro*{printurl}{\iffieldundef{url}{}{\printtext{[PDF: }\href{\thefield{url}}{\nolinkurl{\thefield{url}}}\printtext{]}}}
\DefineBibliographyStrings{english}{
  backrefpage = {Cited p\adddot},
  backrefpages = {Cited pp\adddot}
}
\DeclareBibliographyDriver{misc}{%
  \printnames{author}
  \newunit\newblock
  \printfield{title}
  \newunit\newblock
  \printfield{howpublished}
  \newunit\newblock
  \printfield{year}
  \newunit\newblock
  \usebibmacro{printdoi}
  \usebibmacro{printeprint}
  \usebibmacro{printurl}
  \newunit\newblock
  \usebibmacro{pageref}
  \finentry
}
\patchcmd{\bibsetup}{\interlinepenalty=5000}{\interlinepenalty=10000}{}{}
\usepackage{xurl}
\usepackage[
  colorlinks=true,
  linktocpage=true,
  linkcolor=NavyBlue,
  citecolor=NavyBlue,
  urlcolor=NavyBlue
]{hyperref}
\usepackage[noabbrev,capitalize]{cleveref}

\usepackage{tikzit}
\usetikzlibrary{arrows.meta}
\input{qc.tikzstyles}
\input{qc.tikzdefs}

\theoremstyle{plain}\newtheorem{definition}[thm]{Definition}
\theoremstyle{plain}\newtheorem{theorem}[thm]{Theorem}
\theoremstyle{plain}\newtheorem{proposition}[thm]{Proposition}
\theoremstyle{plain}\newtheorem{lemma}[thm]{Lemma}
\theoremstyle{plain}\newtheorem{corollary}[thm]{Corollary}
\theoremstyle{plain}\newtheorem{example}[thm]{Example}

\usepackage{macros}

\begin{document}

\title[Controlled PROP: Theory and Applications to Quantum Circuits]{Controlled PROP: \\Theory and Applications to Quantum Circuits}
\author[N.~Delorme]{Noé Delorme}
\author[S.~Perdrix]{Simon Perdrix}
\address{Université de Lorraine, CNRS, INRIA, LORIA, 54000 Nancy, France}
\email{noe.delorme@inria.fr, simon.perdrix@inria.fr}

\begin{abstract}
  \noindent Control is a fundamental concept in quantum and reversible computational models. It enables the conditional application of a transformation to a system, depending on the state of another system. We introduce a general framework for diagrammatic reasoning featuring control as a constructor. To this end, we provide an elementary axiomatization of control functors, extending the standard formalism of PROPs to controlled PROPs. As an application, we show that controlled PROPs facilitate diagrammatic reasoning for quantum circuits. In particular, they admit a simple and complete axiomatization involving at most three qubits, whereas in the standard PROP setting, completeness necessarily requires axioms involving arbitrarily many qubits.
\end{abstract}

\maketitle

\tableofcontents

\input{sections}

\section*{Acknowledgements}

The authors would like to thank Alexandre Clément, Emmanuel Jeandel, Louis Lemonnier, William Schober, and Scott Wesley for fruitful discussions. This work is supported by the Plan France 2030 through the PEPR integrated project EPiQ ANR-22-PETQ-0007 and the HQI platform ANR-22-PNCQ-0002; and by the European Union through the MSCA Staff Exchange project Qcomical HORIZON-MSCA-2023-SE-01. The project is also supported by the Maison du Quantique MaQuEst.

\printbibliography

\appendix

\input{appendices}

\end{document}

%% file: qc.tikzstyles
\tikzstyle{nothing}=[shape=rectangle, minimum height=0.5cm]
\tikzstyle{gate}=[fill=bkg, draw=black, shape=rectangle, minimum height=0.43cm, minimum width=0.43cm, inner sep=0.05cm, thick, tikzit fill=white]
\tikzstyle{doublegate}=[fill=bkg, draw=black, shape=rectangle, minimum height=0.93cm, minimum width=0.43cm, inner sep=0.05cm, thick, tikzit fill=white]
\tikzstyle{bcontrol}=[fill=black, draw=black, shape=circle, scale=0.4, thick]
\tikzstyle{wcontrol}=[fill=bkg, draw=black, shape=circle, scale=0.4, thick, tikzit fill=white]
\tikzstyle{bwcontrol}=[draw=black, shape=circle, scale=0.4, path picture={
    \fill[bkg] (path picture bounding box.center) circle(0.5);
    \fill[black] (path picture bounding box.center) [radius=0.5] -- ++(0:0.5) arc[start angle=0, end angle=180, radius=0.5] -- cycle;
}, rotate=45, thick, tikzit fill={rgb,255: red,128; green,128; blue,128}]
\tikzstyle{ccontrol}=[fill=bkg, draw=black, inner sep=0.1em, shape=rounded rectangle, minimum height=0.2cm, thick, tikzit fill=white, font={\tiny}]
\tikzstyle{sbcontrol}=[fill=black, draw=black, shape=rectangle, thick, minimum width=0.14cm, minimum height=0.14cm, inner sep=0pt, outer sep=0pt]
\tikzstyle{swcontrol}=[fill=white, draw=black, shape=rectangle, thick, minimum width=0.14cm, minimum height=0.14cm, inner sep=0pt, outer sep=0pt]
\tikzstyle{sccontrol}=[fill=bkg, draw=black, inner sep=0.1em, shape=rectangle, minimum height=0.2cm, thick, tikzit fill=white, font={\tiny}]
\tikzstyle{not}=[shape=circle, path picture={ 
    \draw[black](path picture bounding box.north) -- (path picture bounding box.south) (path picture bounding box.west) -- (path picture bounding box.east);
}, draw=black, scale=0.9, thick]
\tikzstyle{phase}=[fill=bkg, draw=black, inner sep=0.15em, shape=rounded rectangle, minimum height=0.43cm, thick, tikzit fill=white]
\tikzstyle{winit}=[fill=bkg, draw=black, shape border rotate=-90, regular polygon, regular polygon sides=3, inner sep=0pt, outer sep=0pt, rounded corners=0.2mm, thick, minimum size=0.35cm, tikzit fill=white, xshift=1mm]
\tikzstyle{wterm}=[fill=bkg, draw=black, shape border rotate=90, regular polygon, regular polygon sides=3, inner sep=0pt, outer sep=0pt, rounded corners=0.2mm, thick, minimum size=0.35cm, tikzit fill=white, xshift=-1mm]
\tikzstyle{binit}=[fill=black, draw=black, shape border rotate=-90, regular polygon, regular polygon sides=3, inner sep=0pt, outer sep=0pt, rounded corners=0.2mm, thick, minimum size=0.35cm, xshift=1mm]
\tikzstyle{bterm}=[fill=black, draw=black, shape border rotate=90, regular polygon, regular polygon sides=3, inner sep=0pt, outer sep=0pt, rounded corners=0.2mm, thick, minimum size=0.35cm, xshift=-1mm]
\tikzstyle{ginit}=[fill={rgb,255: red,128; green,128; blue,128}, draw=black, shape border rotate=-90, regular polygon, regular polygon sides=3, inner sep=0pt, outer sep=0pt, rounded corners=0.2mm, thick, minimum size=0.35cm, xshift=1mm]
\tikzstyle{gterm}=[fill={rgb,255: red,128; green,128; blue,128}, draw=black, shape border rotate=90, regular polygon, regular polygon sides=3, inner sep=0pt, outer sep=0pt, rounded corners=0.2mm, thick, minimum size=0.35cm, xshift=-1mm]
\tikzstyle{discard}=[path picture={
    \draw[black,thick](-0.2cm,0)--(-0.12cm,0);
    \draw[black,thick,line cap=round](-0.13cm,-0.15cm)--(-0.13cm,0.15cm) (-0.07cm,-0.11cm)--(-0.07cm,0.11cm) (-0.01cm,-0.07cm)--(-0.01cm,0.07cm) (0.05cm,-0.03cm)--(0.05cm,0.03cm);
}, minimum width=0.4cm, draw=none, outer sep=0pt, minimum height=0.4cm, xshift=-1mm]
\tikzstyle{Zmeasure}=[fill=bkg, draw=black, shape=rectangle, minimum height=0.43cm, minimum width=0.43cm, inner sep=0pt, thick, path picture={
        \draw[thick, ->] (-0.05cm,-0.215cm) -- (0.10cm,0.15cm);
        \draw (path picture bounding box.south west) arc[start angle=180,end angle=0,radius=0.21cm];
        \node at (-0.09cm,+0.09cm) {\scalebox{0.6}{$\bm{Z}$}};
}, rounded corners=1mm]
\tikzstyle{Xmeasure}=[fill=bkg, draw=black, shape=rectangle, minimum height=0.43cm, minimum width=0.43cm, inner sep=0pt, thick, path picture={
        \draw[thick, ->] (-0.05cm,-0.215cm) -- (0.10cm,0.15cm);
        \draw (path picture bounding box.south west) arc[start angle=180,end angle=0,radius=0.21cm];
        \node at (-0.09cm,+0.09cm) {\scalebox{0.6}{$\bm{X}$}};
}, rounded corners=1mm]
\tikzstyle{bar}=[minimum height=0.43cm, minimum width=0.43cm, path picture={
        \draw[thick, line cap=round] (-0.04cm,-0.08cm) -- (0.04cm,0.08cm);
}]
\tikzstyle{arrow}=[-stealth, thick, draw=black]
\tikzstyle{hookarrow}=[right hook-stealth, thick, draw=black]

\tikzstyle{wire}=[-, thick, draw=black, line cap=round]
\tikzstyle{dashed}=[-, draw=black, dash pattern=on 1.6pt off 1.1pt, thick, line cap=round]
\tikzstyle{border}=[-, fill=bkg, draw=black, thick, line cap=round, tikzit fill=white]
\tikzstyle{gwire}=[-, draw={rgb,255: red,128; green,128; blue,128}, thick, line cap=round]
\tikzstyle{gdashed}=[-, draw={rgb,255: red,128; green,128; blue,128}, dash pattern=on 1.6pt off 1.1pt, thick, line cap=round]
\tikzstyle{register}=[-, double, thick]

%% file: sections.tex
\section{Introduction}\label{sec:intro}

\subsection{Controlled PROPs}
Diagrammatic reasoning has been very successful in recent years across various domains, including quantum computing \cite{coeckeduncan,zxcliffordt,zw,zh,abramskycoecke,completeness}, linear algebra \cite{sobocinski,interactingbialgebras}, control theory \cite{signalflowgraphs,baez}, natural language processing \cite{coeckenlp,zengnlp} and symbolic dynamics \cite{jeandelstrongshift} to name a few. The natural mathematical framework for diagrammatic reasoning is that of symmetric monoidal categories, and more precisely, PROPs. This formalism allows representing processes and how they interact in space-time via sequential and parallel compositions of string diagrams.

Control is a central concept in quantum and reversible computational models. Roughly speaking, a transformation can be conditioned on an additional system such that, depending on its state, the transformation is either applied or not. For instance, controlling a NOT gate yields the well-known CNOT gate, and adding another control leads to the Toffoli gate. Notice that various quantum programming languages \cite{quipper,protoquipper,microsoftqsharp,qiskit,cirq} provide primitives that allow controlling gates or even circuits.

In this paper, we introduce a formalization of control as a constructor, i.e.~as a \emph{control functor} in a PROP: every diagram  $f$ with $n$ inputs and $n$ outputs yields a new diagram $\ctrl(f)$ with $n+1$ inputs and $n+1$ outputs, which represents its controlled version. A \emph{controlled PROP} is then defined as a PROP equipped with a control functor. This gives a versatile language that captures the fundamental properties of control and facilitates reasoning about processes. For instance, if an equation $f=g$ is derivable, so is its controlled version $\ctrl(f)=\ctrl(g)$, whereas the latter can be cumbersome to derive when using the standard PROP formalism. Controlled PROPs rely on three constructors -- sequential composition, parallel composition, and control -- and specify how they interact. In particular, control does not preserve parallel composition, i.e. $\ctrl(f\parop g)\neq \ctrl(f)\parop \ctrl (g)$ since the latter acts on one more wire. However, it distributes naturally, as captured by the axioms of control functors.

We also point out a fundamental property playing a central role when handling both control and reversible operations at the same time. Given an arbitrary operation $f$ and a reversible operation $g$, conjugating $f$ by $g$ gives the compute/uncompute pattern ``apply $g$, then $f$, then the inverse of $g$''. This pattern is widely used in algorithm design, for example for a change of basis, for the preparation and release of ancillas, for the implementation of oracles, and for circuit optimization \cite{protoquipper}. The control of such a pattern can be reduced to the control of $f$ alone. This property is known as the \emph{conjugation law} and yields the notion of \emph{conjugated controlled PROP}. This law is graphically represented as follows.
\begin{gather*}
    \tf{conjugation_left}=\tf{conjugation_right}
\end{gather*}

\subsection{Quantum circuits}
An important motivation for our framework comes from quantum computing, where control is extensively used. Complete axiomatizations for quantum circuits have been established in earlier works using the formalism of PROPs \cite{completeness,extensions,minimality}. Completeness ensures that every valid equivalence between quantum circuits can be derived. These results do not rely on control as a constructor, but point out the pivotal role of controlled operations in diagrammatic reasoning, and motivate the introduction of  controlled PROPs. As our main application, we define \emph{controllable quantum circuits} and give a complete axiomatization for them. In its simplest form (see \cref{fig:axuniintro}), where controllable quantum circuits are seen as a conjugated controlled PROP, it contains four simple axioms \labelcref{eq:p_intro,,eq:pplus_intro,,eq:hh_intro,,eq:bg_intro}, together with the standard Euler decomposition axiom \labelcref{eq:e_intro}. We also provide a complete axiomatization for controllable circuits seen as a controlled PROP (not conjugated), with only a few additional relations corresponding to particular instances of the conjugation law.

\begin{figure}[t!]
    \begin{figurebox}\begin{center}
        \vspace{-1em}
        \hspace{-1em}\begin{subfigure}{2.9cm}
            \begin{equation}\label{eq:p_intro}\tag{$\textup{p}$}
                \tf{p_left}=\tf{p_right}
            \end{equation}
        \end{subfigure}\hspace{2em}
        \begin{subfigure}{4.5cm}
            \begin{equation}\label{eq:pplus_intro}\tag{$\textup{p}+$}
                \tf{pplus_left}=\tf{pplus_right}
            \end{equation}
        \end{subfigure}\hspace{2em}
        \begin{subfigure}{4.6cm}
            \begin{equation}\label{eq:hh_intro}\tag{$\textup{hh}$}
                \tf{hh_left}=\tf{hh_right}
            \end{equation}
        \end{subfigure}

        \vspace{1em}
        \begin{subfigure}{4.7cm}
            \begin{equation}\label{eq:bg_intro}\tag{$\textup{bg}$}
                \tf{bgold_03}=\tf{bgold_01}
            \end{equation}
        \end{subfigure}

        \vspace{1em}
        \hspace{-1em}\begin{subfigure}{11.3cm}
            \begin{equation}\label{eq:e_intro}\tag{$\textup{e}$}
                \tf{eprime_var_left}=\tf{eprime_var_right}
            \end{equation}
        \end{subfigure}
    \end{center}\end{figurebox}
\caption{Complete axiomatization for quantum circuits defined as a conjugated controlled PROP generated by the Hadamard gate $\protect\tf{h}$ and global phase gates $\protect\tf{palpha}$, using the notations $\protect\tf{Zalpha}\defeq\ctrl(\!\protect\tf{palpha}\!)$ and $\protect\tf{not}\defeq\protect\tf{notdefz}$. See \cref{sec:cqc} for details. \label{fig:axuniintro}}
\end{figure}

Remarkably, in the axiomatization we introduce for controllable quantum circuits, the axioms involve a bounded number of qubits, whereas it has been shown in \cite{minimality} that any complete axiomatization for (unitary) quantum circuits, defined as a PROP, requires at least one axiom involving $n$ qubits for any $n\in\N$. Hence, controlled PROPs circumvent the unboundedness issue of PROPs, making them particularly well suited to diagrammatic reasoning.

As another application, we consider controllable quantum circuits with auxiliary qubits, a.k.a. \emph{ancillas}, for which we also introduce a simple complete axiomatization. It illustrates that controlled PROPs are not necessarily made up only of endomorphisms (even if the control functor only applies to endomorphisms). Throughout the paper, we also give several examples illustrating the versatility of the formalism, including classical reversible circuits and qudit systems.

\subsection{Related work}
In graphical languages, control is a central notion that has been used in various contexts, but not as a constructor. For quantum circuits, the completeness procedure is based on a back-and-forth translation with optical circuits \cite{completeness} that extensively relies on controlled operations and their basic algebra. In this context, controlled operations are treated as shorthand notation, as they can be decomposed into generators. 

In the ZX-calculus, controlled diagrams have also been extensively used, in particular for defining normal forms \cite{genericnf,controlzxw,fermionsqubits}. Notice that the expressivity of the language yields a slightly different definition of control, allowing the control of arbitrary diagrams (not necessarily with the same number of inputs and outputs). Again, in this context, controlled operations are treated as shorthand notation. 

Another interesting approach is to extend the quantum circuit model with a new kind of composition that graphically consists of vertically connecting two gates, where the behavior of the resulting circuit is a generalized controlled operation. This allows one, for instance, to represent a CNOT by connecting a Z gate and an X gate \cite{schober}. 

Additionally, various quantum programming languages (especially circuit description languages) allow for control of quantum circuits. In particular, Quipper supports control and a \emph{withcompute} primitive, which amounts to conjugating a circuit with another. These features  have been semantically formalized using controllable categories \cite{protoquipper} which involve control functions (not necessarily functors) and \emph{withcompute} as a constructor, to provide a categorical model for high-level quantum programming languages. 

Another quantum programming language introduced  recently in \cite{qcarejustphase} allows arbitrary unitaries to be generated from some simple global phases combined with a quantum analog of the \emph{if let} primitive, which can be embedded in our formalism (see \cref{ex:qcjustphase}). 

$\sqrt\Pi$ \cite{sqrtpi} and Hadamard-$\Pi$ \cite{hadamardpi} are two powerful quantum programming languages for approximately universal fragments of quantum computing, that are equipped with sound and complete equational theories. These languages differ from our approach in several ways: first, they consider fragments of unitary transformations, second, they are not graphical, and third, they rely on the direct sum as a primitive rather than control functors. Indeed, this promising approach based on rig categories relies on two constructors $\parop$ and $\oplus$, corresponding to the tensor product and the direct sum respectively. These two products provide a means of expressing control and enable powerful decompositions of unitary transformations, which, however, cannot always be interpreted as graphical circuit-like transformations. In contrast, control functors categorify the intuitive representation of control in circuits, in particular in the reversible and quantum cases.

Finally, very recently, a construction to turn PROPs with an identified involution into controlled PROPs demonstrated that a control structure is essentially captured by a rig structure \cite{onerig} (see \cref{ex:onerig}). However, this rig-based approach is currently specific to the binary case, whereas controlled PROPs are well suited for more general settings (see \cref{ex:hilbcontrolledprop}).

\section{PROPs}\label{sec:prop}

Graphical languages are defined formally with the formalism of \emph{PROPs}, which are a particular kind of categories admitting a nice representation as \emph{string diagrams}.

\begin{definition}[PROP]
    A \emph{PROP} is a strict symmetric monoidal category whose objects are the natural numbers with addition as tensor product.
\end{definition}

In a nutshell, a PROP $\cat P$ consists of collections of \emph{morphisms} $\cat P(n,m)$ for any objects $n,m\in \mathbb N$ together with two compositions $\seqop$ and $\parop$ satisfying the \emph{coherence laws} depicted in \cref{fig:coherencelawsprop}. We write $f:n\to m$ when $f\in\cat{P}(n,m)$, where $\dom(f)\defeq n$ and $\cod(f)\defeq m$ are respectively called \emph{domain} and \emph{codomain}. Morphisms can be composed \emph{sequentially} ($f\seqop g:n\to m$, for any $f:n\to p$, $g:p\to m$) and in parallel ($f\parop g:n+p\to m+q$, for any $f:n\to m$, $g:p\to q$). A PROP has specific morphisms: identities $\id{n}:n\to n$ for any $n\in\N$ and swaps $\swap{n,m}: n+m\to m+n$ for any $n,m\in\N$. Morphisms admit a nice graphical representation as string diagrams: sequential compositions are depicted horizontally; parallel compositions are depicted vertically; $\id{n}$ is depicted as $n$ wires in parallel; $\swap{n,m}$ is depicted as an exchange of $n$ and $m$ wires; and a morphism $f:n\to m$ is depicted as a box with $n$ input wires and $m$ output wires. Here are some examples.
\begin{gather*}
    \begin{array}{cccccc}
        \id{0} \hspace{1em}&\hspace{1em} \id{1} \hspace{1em}&\hspace{1em} \swap{1,1} \hspace{1em}&\hspace{1em} f \hspace{1em}&\hspace{1em} f\seqop g \hspace{1em}&\hspace{1em} f\parop g\\
        \tf{empty} \hspace{1em}&\hspace{1em} \tf{id} \hspace{1em}&\hspace{1em} \tf{swap} \hspace{1em}&\hspace{1em} \tf{f} \hspace{1em}&\hspace{1em} \tf{fseqg} \hspace{1em}&\hspace{1em} \tf{fparg}
    \end{array}
\end{gather*}

\begin{figure}[t!]
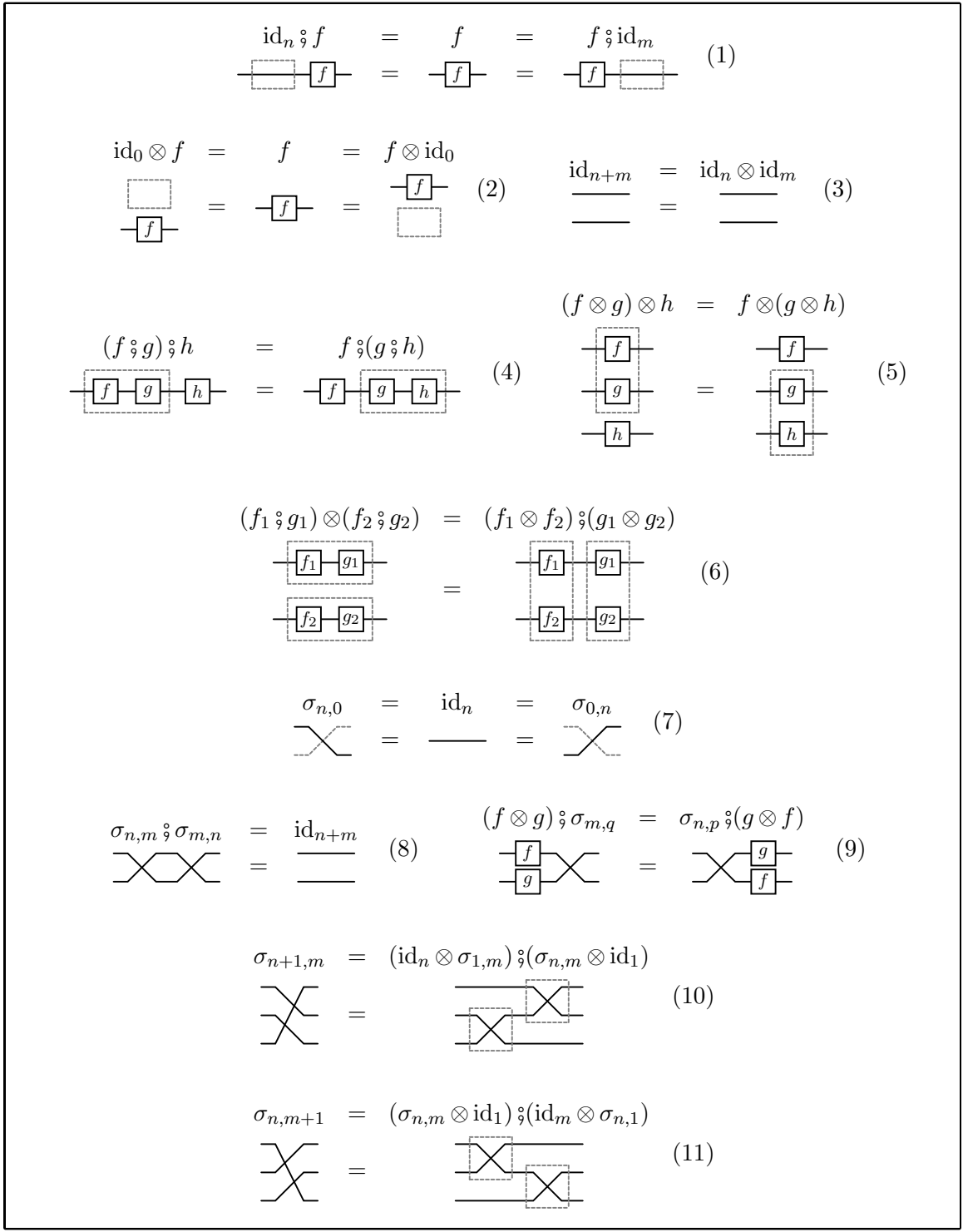

    \begin{figurebox}\begin{center}
        \vspace{-0.5em}
        \hspace{-1em}\begin{subfigure}{8.4cm}
            \begin{equation}\label{eq:seqid}
                \begin{array}{ccccc}
                    \id{n}\seqop f&=&f&=&f\seqop\id{m} \\[0.3em]
                    \tf{seqid_left}&=&\tf{seqid_mid}&=&\tf{seqid_right}
                \end{array}
            \end{equation}
        \end{subfigure}

        \vspace{1.4em}
        \hspace{-1em}\begin{subfigure}{6.6cm}
            \begin{equation}\label{eq:parid}
                \begin{array}{ccccc}
                    \id{0}\parop f&=&f&=&f\parop\id{0} \\[0.3em]
                    \tf{parid_left}&=&\tf{parid_mid}&=&\tf{parid_right}
                \end{array}
            \end{equation}
        \end{subfigure}\hspace{1em}
        \begin{subfigure}{5.0cm}
            \begin{equation}\label{eq:sum}
                \begin{array}{ccc}
                    \id{n+m}&=&\id{n}\parop\id{m} \\[0.3em]
                    \tf{sum_left}&=&\tf{sum_right}
                \end{array}
            \end{equation}
        \end{subfigure}

        \vspace{1.4em}
        \hspace{-1em}\begin{subfigure}{7.7cm}
            \begin{equation}\label{eq:seqasso}
                \begin{array}{ccc}
                    (f\seqop g)\seqop h&=&f\seqop (g\seqop h) \\[0.3em]
                    \tf{seqasso_left}&=&\tf{seqasso_right}
                \end{array}
            \end{equation}
        \end{subfigure}\hspace{0em}
        \begin{subfigure}{6.0cm}
            \begin{equation}\label{eq:parasso}
                \begin{array}{ccc}
                    (f\parop g)\parop h&=&f\parop (g\parop h) \\[0.3em]
                    \tf{parasso_left}&=&\tf{parasso_right}
                \end{array}
            \end{equation}
        \end{subfigure}

        \vspace{1.4em}
        \hspace{-1em}\begin{subfigure}{8.2cm}
            \begin{equation}\label{eq:inter}
                \begin{array}{ccc}
                    (f_1\seqop g_1)\parop(f_2\seqop g_2)&=&(f_1\parop f_2)\seqop(g_1\parop g_2) \\[0.3em]
                    \tf{inter_left}&=&\tf{inter_right}
                \end{array}
            \end{equation}
        \end{subfigure}

        \vspace{1.4em}
        \hspace{-1em}\begin{subfigure}{6.7cm}
            \begin{equation}\label{eq:symid}
                \begin{array}{ccccc}
                    \swap{n,0}&=&\id{n}&=&\swap{0,n} \\[0.3em]
                    \tf{symid_left}&=&\tf{symid_mid}&=&\tf{symid_right}
                \end{array}
            \end{equation}
        \end{subfigure}

        \vspace{1.4em}
        \hspace{-1em}\begin{subfigure}{5.4cm}
            \begin{equation}\label{eq:inv}
                \begin{array}{ccc}
                    \swap{n,m}\seqop\swap{m,n}&=&\id{n+m} \\[0.3em]
                    \tf{inv_left}&=&\tf{inv_right}
                \end{array}
            \end{equation}
        \end{subfigure}\hspace{1em}
        \begin{subfigure}{6.6cm}
            \begin{equation}\label{eq:nat}
                \begin{array}{ccc}
                    (f\parop g)\seqop\swap{m,q}&=&\swap{n,p}\seqop(g\parop f) \\[0.3em]
                    \tf{nat_left}&=&\tf{nat_right}
                \end{array}
            \end{equation}
        \end{subfigure}

        \vspace{1.4em}
        \hspace{-1em}\begin{subfigure}{7.7cm}
            \begin{equation}\label{eq:symtop}
                \begin{array}{ccc}
                    \swap{n+1,m}&=&(\id{n}\parop\swap{1,m})\seqop(\swap{n,m}\parop\id{1}) \\[0.3em]
                    \tf{symtop_left}&=&\tf{symtop_right}
                \end{array}
            \end{equation}
        \end{subfigure}

        \vspace{1.4em}
        \hspace{-1em}\begin{subfigure}{7.7cm}
            \begin{equation}\label{eq:symbot}
                \begin{array}{ccc}
                    \swap{n,m+1}&=&(\swap{n,m}\parop\id{1})\seqop(\id{m}\parop\swap{n,1}) \\[0.3em]
                    \tf{symbot_left}&=&\tf{symbot_right}
                \end{array}
            \end{equation}
        \end{subfigure}
    \end{center}\end{figurebox}
    \caption{\label{fig:coherencelawsprop}Coherence laws of PROPs.}
\end{figure}

When it is clear from the context, we may write \tf{idn}, \tf{idbar}, or even \tf{id} instead of drawing $n$ wires in parallel to simplify diagrams. Moreover, we denote by $\cat{P}_\textup{endo}$ the sub-PROP of endomorphisms of $\cat{P}$, defined by its collections of morphisms $\cat{P}_\textup{endo}(n,n)\defeq\cat{P}(n,n)$ and $\cat{P}_\textup{endo}(n,m)\defeq\varnothing$ when $m\ne n$.

\begin{example}\label{ex:proplinisouni}
    Given a dimension $d\in\N$, let $\catlinmap{d}$ be the PROP of qudit linear maps, meaning that $\catlinmap{d}(n,m)=\linmap{d^n}{d^m}$ is the collection of linear maps from $n$ qudits to $m$ qudits, and where the parallel composition is the usual tensor product. Moreover, let $\catisomap{d}$ and $\catunimap{d}$ be the sub-PROPs of $\catlinmap{d}$ restricted to isometries\footnote{$f:n\to m$ is an isometry if its adjoint $f^\dagger:m\to n$ is a left inverse, meaning $f\seqop f^\dagger=\id{n}$.} and unitaries\footnote{$f:n\to m$ is a unitary if it is an isometry and $n=m$, which implies $f\seqop f^\dagger=f^\dagger\seqop f=\id{n}$} respectively. For instance, $X\defeq\ketbra{1}{0}+\ketbra{0}{1}$ and $Z\defeq\ketbra{0}{0}-\ketbra{1}{1}$ are morphisms in $\catunimap{d}(1,1)$.
\end{example}

Given two PROPs $\cat{P}$ and $\cat{Q}$, a \emph{functor} $F:\cat{P}\to\cat{Q}$ is a map that assigns to each object $n\in\N$ an object $F(n)\in\N$, and to each morphism $f\in \cat{P}(n,m)$ a morphism $F(f)\in\cat{Q}(F(n),F(m))$ while preserving the compositional structure, meaning $F(\id{n})=\id{F(n)}$ and $F(f\seqop g)=F(f)\seqop F(g)$ whenever $f\seqop g$ is defined in $\cat{P}$. An \emph{identity-on-objects functor} is a functor $F$ satisfying $F(n)=n$ for any $n\in\N$. A \emph{monoidal functor} is a functor that additionally preserves the parallel composition, meaning ${F(0)=0}$, ${F(n+m)=F(n)+F(m)}$, and $F(f\parop g)=F(f)\parop F(g)$. And finally, a \emph{PROP functor} is a monoidal functor $F$ that additionally preserves swaps, meaning $F(\swap{n,m})=\swap{F(n),F(m)}$.

A \emph{$\dagger$-PROP} (or \emph{dagger PROP}) is a PROP $\cat{P}$ equipped with an identity-on-objects involutive and contravariant functor $(\cdot)^\dagger:\cat{P}\to\cat{P}$, meaning that each morphism $f:n\to m$ is mapped to a morphism $f^\dagger:m\to n$ such that $(f^\dagger)^\dagger=f$, $(f\seqop g)^\dagger=g^\dagger\seqop f^\dagger$, $(f\parop g)^\dagger=f^\dagger\parop g^\dagger$, and $\swap{n,m}^\dagger=\swap{m,n}$. Moreover, a \emph{$\dagger$-PROP} is said to be a \emph{unitary PROP} when it contains only endomorphisms (meaning $\cat{P}=\cat{P}_\textup{endo}$) and satisfies $f\seqop f^\dagger=f^\dagger\seqop f=\id{n}$ for any $f\in\cat{P}(n,n)$.

A PROP can be specified by generators and equations. A \emph{signature} $\sig$ is a collection of typed \emph{generators} $g:n\to m$ for $n,m\in\N$. We implicitly include the structural morphisms $\id{n}:n\to n$ and $\swap{n,m}:n+m\to m+n$ in $\sig$. The set of \emph{terms} over $\sig$, denoted $\terms{\sig}$, is inductively generated from these generators using sequential and parallel composition: if $f:n\to p$ and $g:p\to m$ are terms, then $f\seqop g:n\to m$ is a term, and if $f:n\to m$ and $g:p\to q$ are terms, then $f\parop g:n+p\to m+q$ is a term. We write $\terms{\sig}(n,m)$ for the set of terms of type $n \to m$. A \emph{context} $K[\cdot]:n\to m$ is a term containing a single hole $[\cdot]:r\to s$, and $K[t]:n\to m$ denotes the term obtained by substituting a well-typed term $t:r\to s$ into this hole. The set of contexts with hole of type $r\to s$ is denoted $\contexts{\terms{\sig}}{r,s}$. An \emph{axiomatization} $\ax$ for $\terms{\sig}$ is a set of well-typed equations $L=R$, called \emph{axioms}, between terms of the same type. An axiom $L=R$ may therefore be applied locally inside any compatible context, replacing $K[L]$ with $K[R]$. Then, \emph{freely generated PROPs} are defined as follows.

\begin{definition}[freely generated PROP]\label{def:freeprop}
    Given a signature $\sig$ and an axiomatization $\ax$ for $\terms{\sig}$, the \emph{freely generated PROP} $\cat{P}\defeq\cat{PROP}_\sig/\ax$ is the PROP whose collection of morphisms $\cat{P}(n,m)$ is the set of terms $\terms{\sig}(n,m)$ quotiented by the coherence laws \labelcref{eq:seqid,eq:parid,eq:sum,eq:seqasso,eq:parasso,eq:inter,eq:symid,eq:inv,eq:nat,eq:symtop,eq:symbot} together with the axioms of $A$. That is, the derivability relation $\cat{P}\vdash \cdot=\cdot$ is the smallest congruence relation on $\terms{\sig}$ containing the equations $\ax\cup\{\text{\labelcref{eq:seqid,eq:parid,eq:sum,eq:seqasso,eq:parasso,eq:inter,eq:symid,eq:inv,eq:nat,eq:symtop,eq:symbot}}\}$, meaning
    \begin{gather*}
        \infer{\cat{P}\vdash t=t}{t\in\terms{\sig}}\;\raisebox{10pt}{,} \hspace{2em}
        \infer{\cat{P}\vdash t_2=t_1}{\cat{P}\vdash t_1=t_2}\;\raisebox{10pt}{,} \hspace{2em}
        \infer{\cat{P}\vdash t_1=t_3}{\cat{P}\vdash t_1=t_2 & \cat{P}\vdash t_2=t_3}\;\raisebox{10pt}{,} \\[0.3em]
        \infer{\cat{P}\vdash L=R}{(L=R)\in\ax\cup\{\text{\labelcref{eq:seqid,eq:parid,eq:sum,eq:seqasso,eq:parasso,eq:inter,eq:symid,eq:inv,eq:nat,eq:symtop,eq:symbot}}\}}\;\raisebox{10pt}{,} \hspace{2em}
        \infer{\cat{P}\vdash K[t_1]=K[t_2]}{\cat{P}\vdash t_1=t_2 & K[\cdot]\in\contexts{\terms{\sig}}{\dom(t_i),\cod(t_i)}}\;\raisebox{10pt}{.}
    \end{gather*}
\end{definition}

Since the relation $\cat{P}\vdash \cdot=\cdot$ is a congruence, the compositions $\seqop$ and $\parop$ are well-defined on equivalence classes: the classes corresponding to $\id{n}$ and $\swap{n,m}$ are the identity and swap morphisms; and the coherence laws \labelcref{eq:seqid,eq:parid,eq:sum,eq:seqasso,eq:parasso,eq:inter,eq:symid,eq:inv,eq:nat,eq:symtop,eq:symbot} hold in the quotient since they belong to the relation. This ensures that $\cat{P}$ is a PROP.

Then, a graphical language can be defined as a freely generated PROP $\cat{P}\defeq\cat{PROP}_\sig/\ax$ equipped with a \emph{standard interpretation}, that is, an identity-on-objects PROP functor $\interp{\cdot}:\cat{PROP}_\sig\to\cat{S}$ where $\cat{S}$ is the PROP providing the semantics of the language. Hence, every diagram $D\in\cat{P}(n,m)$ represents a morphism $\interp{D}\in\cat{S}(n,m)$. The axiomatization $\ax$, and by extension the graphical language $\cat{P}$ is said to be \emph{sound} when every derivable equation holds under the standard interpretation meaning that $\cat{P}\vdash D_1=D_2$ implies $\interp{D_1}=\interp{D_2}$. Conversely, it is \emph{complete} when every semantic equality can be derived diagrammatically, meaning that $\interp{D_1}=\interp{D_2}$ implies $\cat{P}\vdash D_1=D_2$. And finally, it is \emph{universal} when every morphism of $\cat{S}$ can be represented by a diagram of $\cat{P}$, meaning that for any $s\in\cat{S}(n,m)$ there exists $D\in\cat{P}(n,m)$ such that $\interp{D}=s$. Thus, a sound, complete and universal graphical language provides an exact equational presentation of the semantics it represents.

\section{Controlled PROPs}\label{sec:cprop}

\subsection{Control functors}

In this section, we formalize the notion of control as a constructor by equipping a PROP $\cat{P}$ with a \emph{control functor} $\ctrl$ that satisfies a few additional coherence laws (see \cref{fig:coherencelawscprop}). Control functors are only defined on the sub-PROP of endomorphisms $\cat{P}_\textup{endo}$, i.e. on the morphisms with equal domain and codomain.

\begin{definition}[control functor]\label{def:controlfunctor}
    Given a PROP $\cat{P}$, a \emph{control functor} $\ctrl:\cat{P}_\textup{endo}\to\cat{P}_\textup{endo}$ is a functor mapping any object $n\in\N$ to the object $(n+1)$ and mapping any morphism $f:n\to n$ to the morphism $\ctrl(f):(n+1)\to(n+1)$ while satisfying the additional coherence laws \labelcref{eq:strenght,,eq:controlswap,,eq:swapconjugation} depicted in \cref{fig:coherencelawscprop}.
\end{definition}

\begin{figure}[t!]
    \begin{figurebox}\begin{center}
        \vspace{-0.5em}
        \hspace{-1em}\begin{subfigure}{5.1cm}
            \begin{equation}\label{eq:ctrlid}
                \begin{array}{ccc}
                    \ctrl(\id{n})&=&\id{n+1} \\[0.3em]
                    \tf{ctrlid_left}&=&\tf{ctrlid_right}
                \end{array}
            \end{equation}
        \end{subfigure}\hspace{2em}
        \begin{subfigure}{7.0cm}
            \begin{equation}\label{eq:ctrlseq}
                \begin{array}{ccc}
                    \ctrl(f\seqop g)&=&\ctrl(f)\seqop\ctrl(g) \\[0.3em]
                    \tf{ctrlseq_left}&=&\tf{ctrlseq_right}
                \end{array}
            \end{equation}
        \end{subfigure}

        \vspace{1em}
        \hspace{-1em}\begin{subfigure}{6.7cm}
            \begin{equation}\label{eq:strenght}
                \begin{array}{ccc}
                    \ctrl(f\parop \id{m})&=&\ctrl(f)\parop \id{m} \\[0.3em]
                    \tf{strenght_left}&=&\tf{strenght_right}
                \end{array}
            \end{equation}
        \end{subfigure}

        \vspace{1em}
        \hspace{-1em}\begin{subfigure}{10.6cm}
            \begin{equation}\label{eq:controlswap}
                \begin{array}{ccc}
                    (\swap{1,1}\parop \id{n})\seqop\ctrl(\ctrl(f))&=&\ctrl(\ctrl (f))\seqop(\swap{1,1}\parop \id{n}) \\[0.3em]
                    \tf{ctrlswap_left}&=&\tf{ctrlswap_right}
                \end{array}
            \end{equation}
        \end{subfigure}

        \vspace{1em}
        \hspace{-1em}\begin{subfigure}{11.2cm}
            \begin{equation}\label{eq:swapconjugation}
                \begin{array}{ccc}
                    \ctrl(\tau_{k,n,m,\ell}\seqop f\seqop\tau_{k,m,n,\ell}) &=& \tau_{k+1,n,m,\ell}\seqop \ctrl(f)\seqop\tau_{k+1,m,n,\ell} \\[0.3em]
                    \tf{swapconjugation_right}&=&\tf{swapconjugation_left}
                \end{array}
            \end{equation}
        \end{subfigure}
    \end{center}\end{figurebox}
    \caption{Coherence laws of control functors where $\tau_{k,n,m,\ell}\defeq(\id{k}\parop \swap{n,m}\parop \id{\ell})$. \label{fig:coherencelawscprop}}
\end{figure}

Graphically, functors can be depicted using boxes \cite{functorialboxes}. Thus, $\ctrl (f)$ is depicted as a bullet point on the additional wire connected to the diagram $f$, enclosed within a dotted box. This dotted box is omitted  when the diagram is made of a single morphism or when there are nested controls. For instance,
\begin{gather*}
    \begin{array}{cccc}
        \ctrl (f) \hspace{1em}&\hspace{1em} \ctrl(\ctrl (f)) \hspace{1em}&\hspace{1em} \ctrl (f\seqop g) \hspace{1em}&\hspace{1em} \ctrl (f\parop g) \\
        \tf{ctrlf} ,\hspace{1em}&\hspace{1em} \tf{ctrlctrlfvar}\!=\!\tf{ctrlctrlf} ,\hspace{1em}&\hspace{1em} \tf{ctrlfseqg} ,\hspace{1em}&\hspace{1em} \tf{ctrlfparg}.
    \end{array}
\end{gather*}

We may write $\ctrl^n(f)$ for $n$ nested applications of $\ctrl$ from $f$, which can be inductively defined as $\ctrl^0(f)\defeq f$ and $\ctrl^{n+1}(f)\defeq\ctrl(\ctrl^{n}(f))$. Also, the coherence laws of PROPs (see \cref{fig:coherencelawsprop}) allow simplifying the drawings by stretching the line to depict morphisms that are controlled by non-consecutive wires. For instance, the following diagrams are the same modulo the coherence laws.

\begin{equation*}
    \tf{swapctrlf_01}\defeq\tf{swapctrlf_03}
\end{equation*}

On one hand, \labelcref{eq:ctrlid,eq:ctrlseq} hold by definition of the control as a functor. On the other hand, \labelcref{eq:strenght,,eq:controlswap,,eq:swapconjugation} capture the expected behavior of a control: identities can be pulled out from the control; nested controls commute; and swaps, which are left-invertible by definition, satisfy the conjugation law. However, control functors are not monoidal because $\ctrl(f\parop g)\ne\ctrl(f)\parop\ctrl(g)$. Yet, as illustrated by the following derivation, the expected behavior of the control functor applied to a parallel composition is still captured.
    
\begin{gather*}
    \tf{ctrlpar_01}
    \eqeqref{eq:seqid,eq:inter}\tf{ctrlpar_02}
    \eqeqref{eq:inv,eq:nat}\tf{ctrlpar_03}
    \eqeqref{eq:ctrlseq}\tf{ctrlpar_04} \\[0.5em]
    \eqeqref{eq:swapconjugation}\tf{ctrlpar_05} 
    \eqeqref{eq:strenght}\tf{ctrlpar_06}
    \eqeqref{eq:inv,eq:nat}\tf{ctrlpar_07}
\end{gather*}

\begin{example}\label{ex:hilbcontrolledprop}
    In $\catlinmap{2}$ (see \cref{ex:proplinisouni}), the standard notion of control is to apply a linear map $U\in\linmap{2^n}{2^n}$ on the targeted $n$-qubit system when the control qubit is in the state $\ket{1}$, and to apply the identity $\idmat{2^n}$ when the control qubit is in the state $\ket{0}$. This can be embedded as the following control functor.
    \begin{equation*}
        \ctrl_{\ket{1}}:U\mapsto\ketbra{0}{0}\parop \idmat{2^n}+\ketbra{1}{1}\parop U
    \end{equation*}
    However, this is not the only control functor. For instance, 
    \begin{equation*}
        \ctrl_{\ket{0}}:U\mapsto\ketbra{0}{0}\parop U+\ketbra{1}{1}\parop  \idmat{2^n}
        \hspace{1em}\text{and}\hspace{1em}
        \ctrl_{\ket{-}}:U\mapsto\ketbra{+}{+}\parop  \idmat{2^n}+\ketbra{-}{-}\parop U
    \end{equation*}
    are also control functors in $\catlinmap{2}$. In general, in dimension $d\in\N$, qudit linear maps $\catlinmap{d}$ admit a control functor 
    \begin{equation*}
        \textstyle\ctrl_{\ket{k}}: U\mapsto\ketbra{k}{k}\parop U+\sum_{\ell\ne k}\ketbra{\ell}{\ell}\parop  \idmat{d^n}
    \end{equation*}
    for any ${0\le k<d}$. Notice however that the  map ${U\mapsto\sum_{k}\ketbra{k}{k}\parop U^k}$ is not a control functor as it fails to be functorial in general when $d>2$.
\end{example}

\subsection{Controlled PROPs}

We now focus on PROPs equipped with a distinguished control functor, which we call \emph{controlled PROPs}. The general case of PROPs with several control functors is described in \cref{sec:poly}. Notice that a single control functor already offers a rich structure as other control functors may be obtained by conjugating the control wire with some invertible morphism. For instance, $\ctrl_{\ket{0}}$ and $\ctrl_{\ket{-}}$ (see \cref{ex:hilbcontrolledprop}) can be recovered by conjugating the control qubit of $\ctrl_{\ket{1}}$ with $X$ or Hadamard gates respectively.

\begin{definition}[controlled PROP]\label{def:cprop}
    A \emph{controlled PROP (CPROP)} is a PROP equipped with a control functor.
\end{definition}

Unless a different notation is specified, we systematically denote by $\ctrl$ the control functor of a controlled PROP. 

\begin{example}
    $\catlinmap{2}$, $\catisomap{2}$, and $\catunimap{2}$ with $\ctrl_{\ket{1}}$ as control functor are CPROPs.
\end{example}

Any PROP admits two trivial control functors $f\mapsto \id{1}\parop f$ and $f\mapsto \id{n+1}$ (where $f:n\to n$) which correspond to the two degenerate cases of a control, namely when $f$ is respectively always or never applied.  The following notion of \emph{points} can be used to rule out these two degenerate  cases. Intuitively, the points of a control functor are two specific morphisms \emph{false} ($\tf{winit}:0\to 1$) and \emph{true} ($\tf{binit}:0\to 1$) that respectively annihilate and fire the controlled gate when they are plugged to the control wire.

\begin{definition}[points]
    A control functor $\ctrl$ has \emph{points} if there exist two morphisms $\tf{winit}:0\to 1$ and $\tf{binit}:0\to 1$ such that \labelcref{eq:falsepoint,eq:truepoint} are satisfied for any $f:n\to n$.
\end{definition}

\begin{example}
    In $\catlinmap{2}$, $\ket{0}$ and $\ket{1}$ are points of the control functor $\ctrl_{\ket{1}}$. In $\catlinmap{d}$, $\ket{\ell}$ and $\ket{k}$ are points of the control functor $\ctrl_{\ket{k}}$ whenever $\ell\ne k$. Notice however that points do not necessarily exist, in particular when all morphisms of the CPROP are endomorphisms, like in $\catunimap{2}$ for instance.
\end{example}

\begin{figure}[t!]
    \begin{figurebox}\begin{center}
        \vspace{-0.5em}
        \hspace{-1.8em}\begin{subfigure}{7.4cm}
            \begin{equation}\label{eq:falsepoint}
                \begin{array}{ccc}
                    (\tf{winit}\parop \id{n})\seqop\ctrl(f) \!\!\!&=&\!\!\! \tf{winit}\parop \id{n} \\[0.3em]
                    \tf{winitcf} &=& \tf{winitidf}
                \end{array}
            \end{equation}
        \end{subfigure}\hspace{0.3em}
        \begin{subfigure}{7.2cm}
            \begin{equation}\label{eq:truepoint}
                \begin{array}{ccc}
                    (\tf{binit}\parop \id{n})\seqop\ctrl(f) \!\!\!&=&\!\!\! \tf{binit}\parop f \\[0.3em]
                    \tf{binitcf} &=& \tf{binitf}
                \end{array}
            \end{equation}
        \end{subfigure}

        \vspace{1em}
        \hspace{-1em}\begin{subfigure}{9.6cm}
            \begin{equation}\label{eq:conjugation}
                \begin{array}{ccc}
                    \ctrl(g\seqop f\seqop h)&=&(\id{1}\parop g)\seqop \ctrl(f)\seqop (\id{1}\parop h) \\[0.3em]
                    \tf{conjugation_gfh_left}&=&\tf{conjugation_gfh_right}
                \end{array}
            \end{equation}
        \end{subfigure}
    \end{center}\end{figurebox}
    \caption{Additional coherence laws that a controlled PROP may satisfy. \label{fig:morecoherencelawscprop}}
\end{figure}

\begin{definition}[CPROP functor]\label{def:cpropfunctor}
    A \emph{CPROP functor} $F:\cat{P}\to\cat{Q}$ between two CPROPs $\cat{P}$ and $\cat{Q}$ is a PROP functor satisfying $F(\ctrl(f))=\ctrl(F(f))$ for any $f\in\cat{P}_\textup{endo}$.
\end{definition}

Similarly to PROPs (see \cref{sec:prop}), a CPROP $\cat{P}$ can be freely generated by a signature $\sig$ and an axiomatization $\ax$. We extend the notions of terms and contexts to account for the control. Given a signature $\sig$, the set of \emph{controlled terms} over $\sig$, denoted $\cpropterms{\sig}$, is obtained by closing the set of terms under control: whenever $t:n\to n$ is a controlled term, so is $\ctrl(t):n+1\to n+1$. Similarly, \emph{controlled contexts} extend contexts by allowing the hole to occur under control: whenever $K[\cdot]:n\to n$ is a controlled context, so is $\ctrl(K[\cdot]):n+1\to n+1$. We write $\cpropterms{\sig}(n,m)$ for the set of controlled terms of type $n\to m$, and $\contexts{\cpropterms{\sig}}{r,s}$ for the set of controlled contexts whose hole has type $r\to s$. Then, \emph{freely generated CPROPs} are defined as follows.

\begin{definition}[freely generated CPROP]\label{def:freecprop}
    Given a signature $\sig$ and an axiomatization $\ax$ for $\cpropterms{\sig}$, the \emph{freely generated CPROP} $\cat{P}\defeq\cat{CPROP}_\sig/\ax$ is the CPROP whose collection of morphisms $\cat{P}(n,m)$ is the set of controlled terms $\cpropterms{\sig}(n,m)$ quotiented by the coherence laws \labelcref{eq:seqid,eq:parid,eq:sum,eq:seqasso,eq:parasso,eq:inter,eq:symid,eq:inv,eq:nat,eq:symtop,eq:symbot,eq:ctrlid,eq:ctrlseq,eq:strenght,eq:controlswap,eq:swapconjugation} together with the axioms of $A$. That is, the derivability relation $\cat{P}\vdash \cdot=\cdot$ is the smallest congruence relation on $\cpropterms{\sig}$ containing the equations $\ax\cup\{\text{\labelcref{eq:seqid,eq:parid,eq:sum,eq:seqasso,eq:parasso,eq:inter,eq:symid,eq:inv,eq:nat,eq:symtop,eq:symbot,eq:ctrlid,eq:ctrlseq,eq:strenght,eq:controlswap,eq:swapconjugation}}\}$, meaning
    \begin{gather*}
        \infer{\cat{P}\vdash t=t}{t\in\cpropterms{\sig}}\;\raisebox{10pt}{,} \hspace{2em}
        \infer{\cat{P}\vdash t_2=t_1}{\cat{P}\vdash t_1=t_2}\;\raisebox{10pt}{,} \hspace{2em}
        \infer{\cat{P}\vdash t_1=t_3}{\cat{P}\vdash t_1=t_2 & \cat{P}\vdash t_2=t_3}\;\raisebox{10pt}{,} \\[0.5em]
        \infer{\cat{P}\vdash L=R}{(L=R)\in\ax\cup\{\text{\labelcref{eq:seqid,eq:parid,eq:sum,eq:seqasso,eq:parasso,eq:inter,eq:symid,eq:inv,eq:nat,eq:symtop,eq:symbot,eq:ctrlid,eq:ctrlseq,eq:strenght,eq:controlswap,eq:swapconjugation}}\}}\;\raisebox{10pt}{,} \hspace{2em}
        \infer{\cat{P}\vdash K[t_1]=K[t_2]}{\cat{P}\vdash t_1=t_2 & K[\cdot]\in\contexts{\cpropterms{\sig}}{\dom(t_i),\cod(t_i)}}\;\raisebox{10pt}{.}
    \end{gather*}
\end{definition}

Since the relation $\cat{P}\vdash \cdot=\cdot$ is a congruence, the compositions $\seqop$ and $\parop$ and the control $\ctrl$ are well-defined on equivalence classes. Moreover, \labelcref{eq:ctrlid,eq:ctrlseq,eq:strenght,eq:controlswap,eq:swapconjugation} hold in the quotient since they belong to the relation. This ensures that $\cat{P}$ is a CPROP.

Graphical languages can also be defined as freely generated CPROPs. In these settings, the standard interpretation is defined as an identity-on-objects CPROP functor (see \cref{def:cpropfunctor}). Moreover, if an equation $t_1=t_2$ is provable in the language, then its controlled version ${\ctrl(t_1)=\ctrl(t_2)}$ follows directly by construction.

\begin{example}
    Consider the freely generated CPROP whose unique generator is $\tf{not}:1\to 1$ and whose unique axiom is $\tf{NOTNOT_left}=\tf{NOTNOT_right}$. Then, $\raisebox{-1pt}{\tf{smallcnot}}$ is not a generator, but the diagram corresponding to $\ctrl(\tf{not})$; moreover, the equation $\raisebox{-1pt}{\tf{smallcnotcnot_left}}=\raisebox{-1pt}{\tf{smallcnotcnot_right}}$ is derivable.
\end{example}

\subsection{Conjugated controlled PROPs}

Going one step further, we consider \emph{conjugated controlled PROPs}, in which the conjugation law holds.

\begin{definition}[conjugated controlled PROP]
    A \emph{conjugated controlled PROP (CCPROP)} is a CPROP $\cat{P}$ satisfying the \emph{conjugation law} \labelcref{eq:conjugation} for any $g\in\cat{P}(m,n)$, $f\in\cat{P}(n,n)$, and $h\in\cat{P}(n,m)$ such that $g\seqop h=\id{m}$.
\end{definition}

\begin{example}
    $\catlinmap{2}$, $\catisomap{2}$, and $\catunimap{2}$ with $\ctrl_{\ket{1}}$ as control functor are CCPROPs.
\end{example}

The conjugation law is particularly important in algorithm design where the compute/uncompute pattern $g\seqop f\seqop \invers{g}$ (taking $h=\invers{g}$) is often used to prepare and release ancillary bits (or qubits) and implement oracles. This law can also be used in compilation processes to reduce (or extend when applied from right to left) the scope of the control.

Like PROPs and CPROPs (see \cref{def:freeprop,def:freecprop}), a CCPROP $\cat{P}$ can be freely generated by a signature $\sig$ and an axiomatization $\ax$. However, unlike freely generated PROPs and CPROPs, a freely generated CCPROP is not specified by purely equational axioms. This is because the conjugation law \labelcref{eq:conjugation} is conditional, since it applies only when the premise $g\seqop h=\id{m}$ is satisfied. As a consequence, defining $\cat{P}$ and then considering $\cat{P}/\ax$ may create new left-invertible morphisms, and therefore new instances of the conjugation law. Hence, the axioms of $\ax$ and the inference corresponding to the conjugation law must be taken into account simultaneously.

\begin{definition}[freely generated CCPROP]\label{def:freeccprop}
    Given a signature $\sig$ and an axiomatization $\ax$ for $\cpropterms{\sig}$, the \emph{freely generated CCPROP} $\cat{P}\defeq\cat{CCPROP}_\sig/\ax$ is the CCPROP whose collection of morphisms $\cat{P}(n,m)$ is the set of controlled terms $\cpropterms{\sig}(n,m)$ quotiented by the coherence laws \labelcref{eq:seqid,eq:parid,eq:sum,eq:seqasso,eq:parasso,eq:inter,eq:symid,eq:inv,eq:nat,eq:symtop,eq:symbot,eq:ctrlid,eq:ctrlseq,eq:strenght,eq:controlswap,eq:swapconjugation} and the conjugation law \labelcref{eq:conjugation} together with the axioms of $A$. That is, the derivability relation $\cat{P}\vdash \cdot=\cdot$ is the smallest congruence relation on $\cpropterms{\sig}$ containing the equations $\ax\cup\{\text{\labelcref{eq:seqid,eq:parid,eq:sum,eq:seqasso,eq:parasso,eq:inter,eq:symid,eq:inv,eq:nat,eq:symtop,eq:symbot,eq:ctrlid,eq:ctrlseq,eq:strenght,eq:controlswap,eq:swapconjugation}}\}$ and closed under the conjugation law, meaning
    \begin{gather*}
        \infer{\cat{P}\vdash t=t}{t\in\cpropterms{\sig}}\;\raisebox{10pt}{,} \hspace{2em}
        \infer{\cat{P}\vdash t_2=t_1}{\cat{P}\vdash t_1=t_2}\;\raisebox{10pt}{,} \hspace{2em}
        \infer{\cat{P}\vdash t_1=t_3}{\cat{P}\vdash t_1=t_2 & \cat{P}\vdash t_2=t_3}\;\raisebox{10pt}{,} \\[0.5em]
        \infer{\cat{P}\vdash L=R}{(L=R)\in\ax\cup\{\text{\labelcref{eq:seqid,eq:parid,eq:sum,eq:seqasso,eq:parasso,eq:inter,eq:symid,eq:inv,eq:nat,eq:symtop,eq:symbot,eq:ctrlid,eq:ctrlseq,eq:strenght,eq:controlswap,eq:swapconjugation}}\}}\;\raisebox{10pt}{,} \hspace{2em}
        \infer{\cat{P}\vdash K[t_1]=K[t_2]}{\cat{P}\vdash t_1=t_2 & K[\cdot]\in\contexts{\cpropterms{\sig}}{\dom(t_i),\cod(t_i)}}\;\raisebox{10pt}{,} \\[0.5em]
        \infer{\cat{P}\vdash \ctrl(t_1\seqop t_2\seqop t_3)=(\id{1}\parop t_1)\seqop\ctrl(t_2)\seqop(\id{1}\parop t_3)}{t_1:m\to n,t_2:n\to n,t_3:n\to m\in\cpropterms{\sig} & \cat{P}\vdash t_1\seqop t_3=\id{m}}\;\raisebox{10pt}{.}
    \end{gather*}
\end{definition}

\section{Conjugation in unitary controlled PROPs}\label{sec:cond}

A specific class of CPROPs that will be of interest in this paper is when a dagger structure coexists with the control structure, while ensuring that every morphism is invertible. This yields the notion of \emph{unitary CPROP}, extending the notion of \emph{unitary PROP} (see \cref{sec:prop}).

\begin{definition}[unitary CPROP]\label{def:unitarycprop}
    A \emph{unitary CPROP} $\cat{P}$ is a CPROP which is also a unitary PROP satisfying $(\ctrl(f))^\dagger=\ctrl(f^\dagger)$ for any $f\in\cat{P}(n,n)$.
\end{definition}

Notice that, in a unitary CPROP, every morphism is an endomorphism. Moreover, the conjugation law is particularly relevant for unitary CPROPs because, in this context, every inverse of $g$ can be provably transformed into $g^\dagger$. As stated in the following lemma, this allows removing the premise of the conjugation law.

\begin{lemma}\label{lem:conjugationconditionunitary}
    A unitary CPROP $\cat{P}$ is a CCPROP if and only if the following equation \labelcref{eq:conjugationunitary} is satisfied for any $f,g\in\cat{P}(n,n)$.
    \begin{gather}\label{eq:conjugationunitary}
        \ctrl(g\seqop f\seqop g^\dagger)=(\id{1}\parop g)\seqop \ctrl(f)\seqop (\id{1}\parop g^\dagger)
    \end{gather}
\end{lemma}
\begin{proof}
    The instances of \labelcref{eq:conjugationunitary} correspond exactly to the instances of the conjugation law \labelcref{eq:conjugation} where $h=g^\dagger$. To prove that $\cat{P}$ is a conjugated CPROP, it remains to prove the other instances of \labelcref{eq:conjugation}. This is achieved using $\cat{P}\vdash h=g^\dagger\seqop g\seqop h=g^\dagger\seqop \id{n}=g^\dagger$, where the first step uses the unitarity of $\cat{P}$ and the second step uses the hypothesis $\cat{P}\vdash g\seqop h=\id{n}$.
\end{proof}

We provide a sufficient condition for the conjugation law to hold in the case of a unitary CPROP. Intuitively, proving the conjugation law for a small subset of morphisms implies the conjugation law for all of them.

\begin{definition}[reducibility]\label{def:reducibility}
    Given a unitary CPROP $\cat{P}$ and a subset of morphisms $M\subseteq\cat{P}$, let $\cat{P}\vert_M$ be the closure of $M$ under sequential and parallel compositions, where the subsets of morphisms $\cat{P}\vert_M(n,n)$ are inductively defined as
    \begin{gather*}
        \infer{f\in\cat{P}\vert_M(\dom(f),\cod(f))}{f\in M}\;\raisebox{10pt}{,} \hspace{1.5em}
        \infer{f\seqop g\in\cat{P}\vert_M(n,n)}{f,g\in\cat{P}\vert_M(n,n)}\;\raisebox{10pt}{,} \hspace{1.5em}
        \infer{f\parop g\in\cat{P}\vert_M(n+m,n+m)}{f\in\cat{P}\vert_M(n,n) & g\in\cat{P}\vert_M(m,m)}\;\raisebox{10pt}{.}
    \end{gather*}
    Then, we say that $\cat{P}$ is \emph{$M$-reducible} if for any morphism $f\in\cat{P}(n,n)$ there exists a morphism $f'\in\cat{P}\vert_M(n,n)$ such that $\cat{P}\vdash f=f'$.
\end{definition}

To show that an $M$-reducible unitary CPROP is conjugated, it is enough to prove the conjugation law for pairs of morphisms of $M$ that are acting on overlapping wires. In the following, given $N\in\N$ and $h:p\to p$ with $p\le N$, we write $\overline{h}\defeq \id{N-p} \parop h $, the value of $N$ being clear from the context.

\begin{lemma}\label{lem:conjugationcondition}
    An $M$-reducible unitary CPROP $\cat{P}$ is a CCPROP if
    \begin{gather*}
        \cat{P}\vdash\ctrl(\overline{g}\seqop\swap{k,\ell}\seqop\overline{f}\seqop\swap{\ell,k}\seqop\overline{g}^\dagger)=(\id{1}\parop(\overline{g}\seqop\swap{k,\ell}))\seqop\ctrl(\overline{f})\seqop(\id{1}\parop(\swap{\ell,k}\seqop\overline{g}^\dagger)),
    \end{gather*}
    i.e.
    \begin{equation}\label{eq:conjugationcondition}\tag{$\star$}
        \cat{P}\vdash\tf{conjugation_swap_left}=\tf{conjugation_swap_right}
    \end{equation}
    is satisfied for any $f:n\to n$ and $g:m\to m$ in $M$, and any $k,\ell\in\N$ such that $\ell\ge1$ and $\max(n,m)\le k+\ell<n+m$.
\end{lemma}

In other words, \eqref{eq:conjugationcondition} is the conjugation law for $\overline{f}$ conjugated by $\overline{g}\seqop\swap{k,\ell}$. Intuitively, $g$ acts on the last $m$ wires, and $\swap{k,\ell}$ rotates the $N$ wires so that $f$ acts on the wires $k-n+1,\dots,k$ taken cyclically (i.e.~modulo $N$). The constraint $k+\ell<n+m$ ensures that $f$ and $g$ overlap, and every way of placing $f$ and $g$ on overlapping wires is obtained, up to a permutation of the wires, for some $k,\ell$ (see the proof below). In particular, there is nothing to check when $n=0$ or $m=0$ (e.g.~for global phases).

\begin{proof}
   For any $f,g,h\in\cat{P}(n,n)$ satisfying $\cat{P}\vdash g\seqop h=\id{n}$, we prove that the conjugation law holds, meaning
    \begin{gather*}
        \cat{P}\vdash \tf{conjugation_biggfh_left} = \tf{conjugation_biggfh_right}.
    \end{gather*}
 Notice that $\cat P \vdash h = g^\dagger \seqop g \seqop h = g^\dagger \seqop \id{n} = g^\dagger$, so it is enough to prove, for any $f,g\in\cat{P}(n,n)$ that 
     \begin{gather}\label{eq:conjugationprime}
        \tf{conjugation_big_left}=\tf{conjugation_big_right}
    \end{gather}

  As $\cat{P}$ is $M$-reducible, any morphism  $h\in \cat{P}(n,n)$ can be transformed into a morphism  $h'\in \cat{P}\vert_{M}(n,n)$.
    Moreover, notice that the coherence laws of PROPs (see \cref{fig:coherencelawsprop}) ensure that there exist $\lambda\ge1$ and $p_j,q_j\in\N$, $h_j:n_j\to n_j\in M$ satisfying $p_j+n_j+q_j=n$ for all $1\le j\le \lambda$, such that
    \begin{gather*}
        \cat{P}\vdash h = (\id{p_1}\parop h_1\parop\id{q_1})\seqop(\id{p_{2}}\parop h_{2}\parop\id{q_{2}})\seqop\dots\seqop(\id{p_{\lambda}}\parop h_{\lambda}\parop\id{q_{\lambda}}).
    \end{gather*}
    Intuitively, we can transform any morphism $h$ into a sequential composition of $\lambda$ layers each containing a single morphism of $M$ sandwiched between identities. Let $\lambda_h$ be the minimal number of such layers. Notice that for any $h\in \cat{P}(n,n)$, if $\lambda_h>1$ then  there exist $h_1, h_2$ such that $\cat P \vdash h=h_1\seqop h_2$ and $\max(\lambda_{h_1},\lambda_{h_2})<\lambda_h$.  
 We thus prove \eqref{eq:conjugationprime} by induction on $(\lambda_g,\lambda_f)$ with the lexicographic order. 

\begin{itemize}
\item For any $f\in \cat{P}(n,n)$ with $\lambda_f>1$, let $f_1, f_2\in \cat{P}(n,n)$ such that $f=f_1\seqop f_2$ and $\max(\lambda_{f_1},\lambda_{f_2})<\lambda_f$. 
   \begin{gather*}
        \tf{conjugation_big_left}
        =\tf{conjproof_f_01-} \\[0.5em]
        \eqcref{def:unitarycprop}\tf{conjproof_f_02-} 
        \eqeqref{eq:ctrlseq}\tf{conjproof_f_03-} \\[0.5em]
        \eqihref\tf{conjproof_f_04-} \\[0.5em]
        \eqcref{def:unitarycprop}\tf{conjproof_f_05-} 
        \eqeqref{eq:ctrlseq}\tf{conjproof_f_06-}
        =\tf{conjugation_big_prime_right-}
    \end{gather*}
    where the induction hypothesis is applied for $(\lambda_g,\lambda_{f_i})<(\lambda_g,\lambda_{f})$. 
    \item For any  $g\in \cat{P}(n,n)$ with $\lambda_g>1$, let $g_1, g_2\in \cat{P}(n,n)$ such that $g=g_1\seqop g_2$ and $\max(\lambda_{g_1},\lambda_{g_2})<\lambda_g$.
    \begin{gather*}
        \tf{conjugation_big_left}
        =\tf{conjproof_g_01-} \\[0.5em]
        \eqihref\tf{conjproof_g_02-} \\[0.5em]
        \eqihref\tf{conjproof_g_03-}
        =\tf{conjugation_big_right}
    \end{gather*}
where the first induction hypothesis is applied for $(\lambda_{g_1},\lambda_{f}+\lambda_{g_2}+\lambda_{g_2^\dagger})<(\lambda_g,\lambda_{f})$, and the second one for $(\lambda_{g_2},\lambda_{f})<(\lambda_g,\lambda_{f})$. 
\item When $\lambda_f=\lambda_g=1$, 
$\cat P\vdash f = \id{p} \otimes f'\otimes \id{q}$ and $\cat P\vdash g = \id{r} \otimes g'\otimes \id{s}$, with $f':n\to n, g':m\to m \in M$.  W.l.o.g. we assume $\min(q,s)=\min(p,r)=0$, as one can reduce to this case by pushing identities out of the scope of the control using \eqref{eq:strenght} and  \eqref{eq:swapconjugation}. Let $N \defeq p+n+q = r+m+s$, if $N\ge n+m$ then $f$ and $g$ act on distinct wires and commute trivially, thus we assume $N<n+m$. Since $\min(q,s)=0$, either $q=0$, or $q\ge1$ and $s=0$. 
\subitem{--} If $q=0$ (which includes the case $q=s=0$), then $\cat P\vdash g  = \sigma_{N-s,s} \seqop \overline{g'} \seqop \sigma_{s,N-s}$ and $\cat P\vdash f = \overline {f'}$. Using \eqref{eq:swapconjugation} to move the outer swaps out of the scope of the control, and then \eqref{eq:conjugationcondition} with $k=s$ and $\ell=N-s$ (notice that $\ell=r+m\ge1$), we get
 \begin{gather*}
        \tf{conjugation_big_left}
        =\tf{conjugation_swap_q0_1} \\[0.5em]
        \eqeqref{eq:swapconjugation}\tf{conjugation_swap_q0_2} \\[0.5em]
        \eqeqref{eq:conjugationcondition}\tf{conjugation_swap_q0_3}
        =\tf{conjugation_big_prime_right-}
    \end{gather*}
\subitem{--} If $q\ge1$, then $s=0$, so $\cat P\vdash g  = \overline{g'}$ and $\cat P\vdash f = \sigma_{N-q,q} \seqop  \overline {f'} \seqop \sigma_{q,N-q}$. Using \eqref{eq:conjugationcondition} with $k=N-q$ and $\ell=q\ge1$, and then \eqref{eq:swapconjugation} to move the swaps into the scope of the control, we get
 \begin{gather*}
        \tf{conjugation_big_left}
        =\tf{conjugation_swap_s0_1} \\[0.5em]
        \eqeqref{eq:conjugationcondition}\tf{conjugation_swap_s0_2}
        \eqeqref{eq:swapconjugation}\tf{conjugation_swap_s0_3}  \\[0.5em]
         =\tf{conjugation_big_prime_right-}\qedhere
    \end{gather*}
\end{itemize}
\end{proof}

\section{Vanilla quantum circuits}\label{sec:qc}

In this section, following  \cite{completeness,extensions,minimality}, we consider the usual definition of quantum circuits used to represent exactly all unitary maps in $\catunimap{2}$ (see \cref{ex:proplinisouni}). That is, the graphical language of \emph{vanilla quantum circuits} is defined as the freely generated PROP (see \cref{def:freeprop}) whose signature $\siguni$ contains the four kinds of generators
\begin{gather*}
    \tf{palpha}:0\to 0, \hspace{2em}
    \tf{Zalpha}:1\to 1, \hspace{2em}
    \tf{H}:1\to 1, \hspace{2em}
    \tf{CNOT}:2\to 2,
\end{gather*}
where the first two are parametrized by $\alpha\in\R$. All generators correspond to usual quantum gates, namely global phases, Z-rotations, Hadamard, and CNOT. Moreover, the axiomatization $\axuni$ of vanilla quantum circuits is depicted in \cref{fig:axuni}.

\begin{definition}[vanilla quantum circuits]\label{def:qc}
    Let $\catuni\defeq\cat{PROP}_{\siguni}/\axuni$ be the graphical language of \emph{vanilla quantum circuits} whose standard interpretation $\intuni{\cdot}:\catuni\to\catunimap{2}$ is the identity-on-objects PROP functor inductively defined as
    \begin{gather*}
        \intuni{\tf{palpha}}=e^{i\alpha}, \hspace{2em}
        \intuni{\tf{Zalpha}}=\ketbra{0}{0}+e^{i\alpha}\ketbra{1}{1}, \\[0.5em]
        \intuni{\tf{H}}=\ketbra{+}{0}+\ketbra{-}{1}, \hspace{2em}
        \textstyle\intuni{\tf{CNOT}}=\ketbra{00}{00}+\ketbra{01}{01}+\ketbra{10}{11}+\ketbra{11}{10}.
    \end{gather*}
\end{definition}

\begin{proposition}\label{prop:qcuniversality}
    $\catuni$ is universal.
\end{proposition}
\begin{proof}
    The well-known Euler decomposition states that Z- and X-rotations (which can be implemented as \tf{Zalpha} and \tf{Xalphadef}) together with global phases ensure that arbitrary one-qubit unitaries can be implemented in $\catuni$. Moreover, any entangling two-qubit gate, such as the CNOT gate, together with arbitrary one-qubit unitaries is universal for unitaries on $n$ qubits \cite{universality,nielsenchuang}.
\end{proof}

Each axiom of $\axuni$ is simple and admits a meaningful interpretation: \labelcref{eq:p,,eq:pplus} show that global phases form a group and their parameters are $2\pi$-periodic; \labelcref{eq:z,,eq:zplus} show that Z-rotations form a group and their parameters are $2\pi$-periodic; \labelcref{eq:hh} shows that the Hadamard gate is self-inverse; \labelcref{eq:bg} shows that a swap can be implemented with three CNOTs; \labelcref{eq:cx} shows that the CNOT gate is self-inverse (when $\alpha=0$) and that Z-rotations commute on the control qubit of a CNOT; \labelcref{eq:cz} shows two different implementations of the controlled-$Z$ gate, using either one or two CNOTs; and finally \labelcref{eq:zn} is a generalization of \labelcref{eq:z} where the circuit $Z^n(\alpha)\in\catuni(n+1,n+1)$ ($n\in\N$ and $\alpha\in\R$) is inductively defined as $Z^0(\alpha)\defeq\tf{Zalpha}$ and
\begin{gather*}
    \tf{lambda_def_left}\defeq\tf{lambda_def_right}.
\end{gather*}

Intuitively, $Z^n(\alpha)$ is an implementation of an $n$-controlled Z-rotation with parameter $\alpha\in\R$, which means $\intuni{Z^n(\alpha)}=e^{i\alpha}\ketbra{1\dots1}{1\dots1}+\sum_{x\ne1\dots1}\ketbra{x}{x}$. As a consequence, each instance of \labelcref{eq:zn} is semantically trivial ($\intuni{Z^n(2\pi)}=\idmat{2^{n+1}}$) but syntactically involved as $Z^n(2\pi)$ contains a number of \tf{Z2piover2n} gates which is exponential in $n$. 

Surely, the most powerful axiom of $\axuni$ is \labelcref{eq:e}, which has first been introduced in the context of ZX-calculus \cite{zx2qubitscliffordt,nearminimalzx}. It follows from the well-known Euler decomposition stating that any one-qubit unitary can be decomposed, up to a global phase, into X- and Z-rotations. It also represents an infinite family of equations where the parameters $\alpha_1,\alpha_2$ are arbitrary in $\R$ and the parameters $\beta_0,\beta_1,\beta_2,\beta_3$ are restricted to $[0,2\pi)$ and computed with a map $\euler:\R^2\to[0,2\pi)^4$ such that
\begin{gather*}
    \euler(\alpha_1,\alpha_2)=(\beta_0,\beta_1,\beta_2,\beta_3)=(b_0\bmod2\pi,b_1\bmod2\pi,b_2\bmod2\pi,b_3\bmod2\pi)
\end{gather*}
where $b_0,b_1,b_2,b_3\in\R$ are computed as follows using the intermediate complex numbers $u,v\in\C$.
\begin{gather*}
        u\defeq -\sin\left(\frac{\alpha_1+\alpha_2}{2}\right)+\ii\cos\left(\frac{\alpha_1-\alpha_2}{2}\right) \hspace{2em}
        v\defeq \cos\left(\frac{\alpha_1+\alpha_2}{2}\right)-\ii\sin\left(\frac{\alpha_1-\alpha_2}{2}\right)
    \end{gather*}
\begin{center}
    \begin{tabular}{c|c|c|c|c}
        & $b_1$ & $b_2$ & $b_3$ & $b_0$ \\
        \hline
        if $v=0$
            & $2\arg(u)$
            & $0$
            & $0$
            & \multirow{3}{*}{$\frac{(\alpha_1+\alpha_2-b_1-b_2-b_3+\pi)}{2}$} \\
        \cline{1-4}
        if $u=0$
            & $2\arg(v)$
            & $\pi$
            & $0$
            & \\
        \cline{1-4}
        otherwise
            & $\arg(u)+\arg(v)$
            & $2\arg(\ii+\left\lvert\nicefrac{u}{v}\right\rvert)$
            & $\arg(u)-\arg(v)$
            &
    \end{tabular}
\end{center}

\begin{figure}[t!]
    \begin{figurebox}\begin{center}
        \vspace{-0.5em}
        \hspace{-1em}\begin{subfigure}{2.9cm}
            \begin{equation*}\label{eq:p}\tag{$\textup{p}$}
                \tf{p_left}=\tf{p_right}
            \end{equation*}
        \end{subfigure}\hspace{2em}
        \begin{subfigure}{4.4cm}
            \begin{equation*}\label{eq:pplus}\tag{$\textup{p}+$}
                \tf{pplus_left}=\tf{pplus_right}
            \end{equation*}
        \end{subfigure}\hspace{2em}
        \begin{subfigure}{4.6cm}
            \begin{equation*}\label{eq:hh}\tag{$\textup{hh}$}
                \tf{hh_left}=\tf{hh_right}
            \end{equation*}
        \end{subfigure}

        \vspace{1em}
        \hspace{-1em}\begin{subfigure}{4.3cm}
            \begin{equation*}\label{eq:z}\tag{$\textup{z}$}
                \tf{z_left}=\tf{z_right}
            \end{equation*}
        \end{subfigure}\hspace{3em}
        \begin{subfigure}{6.6cm}
            \begin{equation*}\label{eq:zplus}\tag{$\textup{z}+$}
                \tf{Zplus_left}=\tf{Zplus_right}
            \end{equation*}
        \end{subfigure}

        \vspace{1em}
        \hspace{-1em}\begin{subfigure}{11.3cm}
            \begin{equation*}\label{eq:e}\tag{$\textup{e}$}
                \tf{eprime_var_left}=\tf{eprime_var_right}
            \end{equation*}
        \end{subfigure}

        \vspace{1em}
        \hspace{-1em}\begin{subfigure}{4.7cm}
            \begin{equation*}\label{eq:bg}\tag{$\textup{bg}$}
                \tf{bg_left}=\tf{bg_right}
            \end{equation*}
        \end{subfigure}\hspace{3em}
        \begin{subfigure}{5.3cm}
            \begin{equation*}\label{eq:cx}\tag{$\textup{cx}$}
                \tf{cx_left}=\tf{cx_right}
            \end{equation*}
        \end{subfigure}

        \vspace{1em}
        \hspace{-1em}\begin{subfigure}{7.7cm}
            \begin{equation*}\label{eq:cz}\tag{$\textup{cz}$}
                \tf{cz_left}=\tf{cz_right}
            \end{equation*}
        \end{subfigure}\hspace{3em}
        \begin{subfigure}{4.8cm}
            \begin{equation*}\label{eq:zn}\tag{$\textup{z}_n$}
                \tf{zn_var_left}=\tf{zn_var_right}
            \end{equation*}
        \end{subfigure}
    \end{center}\end{figurebox}
    \caption{Axiomatization $\axuni$ for $\catuni$, containing an instance of \labelcref{eq:zn} for any $n\ge2$. \label{fig:axuni}}
\end{figure}

\begin{proposition}\label{prop:qcsoundness}
    The axiomatization $\axuni$ is sound.
\end{proposition}
\begin{proof}
    It is enough to check that $\intuni{L}=\intuni{R}$ for any $(L=R)\in\axuni$, which is done by direct computation. Only the case of \labelcref{eq:e} is not straightforward and as been treated in \cite{zx2qubitscliffordt}.
\end{proof}

The first completeness result for quantum circuits has been established in \cite{completeness}; and the axiomatization $\axuni$ has been obtained after significant simplifications in \cite{extensions,minimality}.

\begin{theorem}[\cite{minimality}]\label{th:qccompleteness}
    $\catuni$ is complete.
\end{theorem}

It has been shown in \cite{minimality} that every axiom of $\axuni$ is necessary, which means that none can be removed without losing completeness. Notice that there is an instance of \labelcref{eq:zn} for any number of control qubits $n\ge2$. All such instances are necessary. More generally, any axiomatization yielding a universal and complete graphical language (defined as a PROP) for unitary maps requires at least one axiom involving $n$ qubits for any $n\in\N$. As we will see in the next section, the use of controlled PROPs simplifies the diagrammatic reasoning on quantum circuits by circumventing the necessity of axioms acting on an unbounded number of qubits.

\section{Controllable quantum circuits}\label{sec:cqc}

\subsection{Controllable quantum circuits}

In this section, we define \emph{controllable quantum circuits} as a freely generated CPROP (see \cref{def:freecprop}). We prove that the conjugation law holds in these settings, and provide a complete axiomatization. 

Using the CPROP formalism, we can use much simpler generators to define quantum circuits. This is because a Z-rotation can be seen as a controlled global phase, and a CNOT can be seen as a controlled NOT. In fact, we can simply take the global phases together with the Hadamard gate as generators and still capture all unitaries in $\catunimap{2}$. That is, the graphical language of controllable quantum circuits is defined as the freely generated CPROP whose signature $\sigcuni$ contains the two kinds of generators
\begin{gather*}
    \tf{palpha}:0\to 0
    \hspace{2em}\text{and}\hspace{2em}
    \tf{H}:1\to 1
\end{gather*}
where the first is parametrized by $\alpha\in\R$. Moreover, the axiomatization $\axcuni$ of controllable quantum circuits is depicted in \cref{fig:axcuni}.

\begin{definition}[controllable quantum circuits]\label{def:cqc}
    Let $\catcuni\defeq\cat{CPROP}_{\sigcuni}/\axcuni$ be the graphical language of \emph{controllable quantum circuits} whose standard interpretation $\intcuni{\cdot}:\catcuni\to\catunimap{2}$ is the identity-on-objects CPROP functor inductively defined as
    \begin{gather*}
        \intcuni{\ctrl(C)}\defeq\ctrl_{\ket{1}}(\intcuni{C}), \hspace{2em}
        \intcuni{\tf{palpha}}\defeq\intuni{\tf{palpha}}, \hspace{2em}
        \intcuni{\tf{H}}\defeq\intuni{\tf{H}}.
    \end{gather*}
\end{definition}

\begin{proposition}\label{prop:cqcuniversality}
    $\catcuni$ is universal.
\end{proposition}
\begin{proof}
    Direct consequence of \cref{prop:qcuniversality} because
    \begin{gather*}
        \intuni{\tf{Zalpha}}=\intcuni{\tf{Zalphadef}}
        \hspace{2em}\text{and}\hspace{2em}
        \intuni{\tf{CNOT}}=\intcuni{\tf{CNOTdef_middle}}.\qedhere
    \end{gather*}
\end{proof}

The axiomatization $\axcuni$ contains two kinds of axioms. On the one hand, \labelcref{eq:p,,eq:pplus,,eq:hh,,eq:cbg,,eq:ce} directly correspond to some axioms in $\axuni$. In particular the parameters of \labelcref{eq:ce} are computed with the same functions as in \labelcref{eq:e}. On the other hand, \labelcref{eq:conjhcalpha,,eq:conjcalphah,,eq:conjhccpi,,eq:conjccpih} are new and can be seen as instances of the conjugation law. Notice that some axioms of $\axuni$ do not appear in $\axcuni$ as they are trivialized within the CPROP formalism.

\begin{figure}[t!]
    \begin{figurebox}\begin{center}
        \vspace{-0.5em}
        \hspace{-1em}\begin{subfigure}{2.9cm}
            \begin{equation*}\tag{$\textup{p}$}
                \tf{p_left}=\tf{p_right}
            \end{equation*}
        \end{subfigure}\hspace{2em}
        \begin{subfigure}{4.5cm}
            \begin{equation*}\tag{$\textup{p}+$}
                \tf{pplus_left}=\tf{pplus_right}
            \end{equation*}
        \end{subfigure}\hspace{2em}
        \begin{subfigure}{4.6cm}
            \begin{equation*}\tag{$\textup{hh}$}
                \tf{hh_left}=\tf{hh_right}
            \end{equation*}
        \end{subfigure}

        \vspace{1em}
        \hspace{-1em}\begin{subfigure}{7.5cm}
            \begin{equation}\label{eq:cbg}\tag{$\textup{cbg}$}
                \tf{cbg_left}=\tf{cbg_right}
            \end{equation}
        \end{subfigure}

        \vspace{0.5em}
        \hspace{-1em}\begin{subfigure}{9.2cm}
            \begin{equation}\label{eq:ce}\tag{$\textup{ce}$}
                \tf{ce_left}=\tf{ce_right}
            \end{equation}
        \end{subfigure}

        \vspace{0.6em}\noindent\rule{\textwidth}{0.4pt}
  
        \vspace{0.5em}
        \hspace{-1em}\begin{subfigure}{6.5cm}
            \begin{equation}\label{eq:conjhcalpha}\tag{$\textup{ch}_1$}
                \tf{conjhcalpha_left}=\tf{conjhcalpha_right}
            \end{equation}
        \end{subfigure}\hspace{2em}
        \begin{subfigure}{6.5cm}
            \begin{equation}\label{eq:conjcalphah}\tag{$\textup{ch}_2$}
                \tf{conjcalphah_left}=\tf{conjcalphah_right}
            \end{equation}
        \end{subfigure}
        
        \vspace{0.5em}
        \hspace{-1em}\begin{subfigure}{6.5cm}
            \begin{equation}\label{eq:conjhccpi}\tag{$\textup{ch}_3$}
                \tf{conjhcccpi_left}=\tf{conjhcccpi_right}
            \end{equation}
        \end{subfigure}\hspace{2em}
        \begin{subfigure}{6.5cm}
            \begin{equation}\label{eq:conjccpih}\tag{$\textup{ch}_4$}
                \tf{conjccpich_left}=\tf{conjccpich_right}
            \end{equation}
        \end{subfigure}
    \end{center}\end{figurebox}
    \caption{The axioms \labelcref{eq:p,,eq:pplus,,eq:hh,,eq:cbg,,eq:ce,,eq:conjhcalpha,,eq:conjcalphah,,eq:conjhccpi,,eq:conjccpih} form the complete axiomatization $\axcuni$ for $\catcuni$ (where quantum circuits are defined as a CPROP); and the subset of axioms \labelcref{eq:p,,eq:pplus,,eq:hh,,eq:cbg,,eq:ce} form the complete axiomatization $\axccuni$ for $\catccuni$ (where quantum circuits are defined as a CCPROP).\label{fig:axcuni}}
\end{figure}

\begin{proposition}\label{prop:cqcsoundness}
    The axiomatization $\axcuni$ is sound.
\end{proposition}
\begin{proof}
    Similarly to \cref{prop:qcsoundness}, it is enough to check that $\intcuni{L}=\intcuni{R}$ for any axiom $(L=R)\in\axcuni$, which is done by direct computation.
\end{proof}

We proceed in two steps: first, we prove that $\catcuni$ is a conjugated CPROP using the sufficient condition established in \cref{lem:conjugationcondition}; and then, we prove that $\catcuni$ is complete using an encoding/decoding procedure and the completeness of vanilla quantum circuits (see \cref{th:qccompleteness}). To this end, we prove the equations depicted in \cref{fig:cqcidentities}, whose proofs are given in order of dependency in \cref{app:cqcidentities}.

\begin{figure}[t!]
    \begin{figurebox}\begin{center}
        \vspace{-0.5em}
        \hspace{-1em}\begin{subfigure}{3cm}
            \begin{equation}\label{eq:p0}
                \tf{p0_01}=\tf{p0_04}
            \end{equation}
        \end{subfigure}\hspace{2em}
        \begin{subfigure}{7.5cm}
            \begin{equation}\label{eq:cph}
                \tf{heuler_01}=\tf{heuler_03}
            \end{equation}
        \end{subfigure}

        \vspace{1em}
        \hspace{-1em}\begin{subfigure}{8.1cm}
            \begin{equation}\label{eq:calphaminus}
                \tf{xcalpha_01}=\tf{xcalpha_06}
            \end{equation}
        \end{subfigure}

        \vspace{1em}
        \hspace{-1em}\begin{subfigure}{8.9cm}
            \begin{equation}\label{eq:cphvar}
                \tf{heulervar_01}=\tf{heulervar_05}
            \end{equation}
        \end{subfigure}

        \vspace{1em}
        \hspace{-1em}\begin{subfigure}{8.6cm}
            \begin{equation}\label{eq:chreduce}
                \tf{chaltdef_01}=\tf{chaltdef_11}
            \end{equation}
        \end{subfigure}

        \vspace{1em}
        \hspace{-1em}\begin{subfigure}{8.9cm}
            \begin{equation}\label{eq:ccalphareduce}
                \tf{ccalphadef_01}=\tf{ccalphadef_10}
            \end{equation}
        \end{subfigure}
        
        \vspace{1em}
        \hspace{-1em}\begin{subfigure}{8.8cm}
            \begin{equation}\label{eq:cccpireduce}
                \tf{cccpidef_01}=\tf{cccpidef_10}
            \end{equation}
        \end{subfigure}

        \vspace{1em}
        \hspace{-1em}\begin{subfigure}{6.4cm}
            \begin{equation}\label{eq:conjhh}
                \tf{conjhh_01}=\tf{conjhh_09}
            \end{equation}
        \end{subfigure}

        \vspace{1em}
        \hspace{-1em}\begin{subfigure}{6.5cm}
            \begin{equation}\label{eq:conjhccpivar}
                \tf{conjhccpi_04}=\tf{conjhccpi_01}
            \end{equation}
        \end{subfigure}\hspace{2em}
        \begin{subfigure}{6.5cm}
            \begin{equation}\label{eq:conjccpihvar}
                \tf{conjccpih_04}=\tf{conjccpih_01}
            \end{equation}
        \end{subfigure}
        
    \end{center}\end{figurebox}
    \caption{Some equations provable in $\catcuni$. \label{fig:cqcidentities}}
\end{figure}

\subsection{The conjugation law}

To prove that $\catcuni$ is a conjugated CPROP, we use the sufficient condition of \cref{lem:conjugationcondition}. To do so, we first prove that $\catcuni$ is a unitary CPROP.

\begin{lemma}\label{lem:cqcunitary}
    $\catcuni$ is a unitary CPROP where $(\cdot)^\dagger:\catcuni\to\catcuni$ is defined as
    \begin{gather*}
        (C_1\seqop C_2)^\dagger\defeq C_2^\dagger\seqop C_1^\dagger, \hspace{2em}
        (C_1\parop C_2)^\dagger\defeq C_1^\dagger\parop C_2^\dagger, \hspace{2em}
        (\ctrl(C))^\dagger\defeq \ctrl(C^\dagger), \\[0.5em]
        \id{n}^\dagger\defeq \id{n}, \hspace{2em}
        \swap{n,m}^\dagger\defeq \swap{m,n}, \hspace{2em}
        \left(\tf{palpha}\right)^\dagger\defeq\tf{minusalpha}, \hspace{2em}
        \left(\tf{H}\right)^\dagger\defeq\tf{H}.
    \end{gather*}
\end{lemma}
\begin{proof}
    We can prove $\catcuni\vdash C\seqop C^\dagger=C^\dagger\seqop C=\id{n}$ straightforwardly by induction on $C$, where the base case $C=\tf{alpha}$ and $C=\tf{H}$ are directly given by \labelcref{eq:p0,,eq:pplus,,eq:hh}.
\end{proof}

In order to apply \cref{lem:conjugationcondition}, we additionally show that $\catcuni$ is $M_+$-reducible and $M$-reducible where $M$ and $M_+$ are the sets of morphisms defined as
\begin{gather*}
    \textstyle M\defeq
    \cup_{\alpha\in\R}\left\{\tf{h}\;,\;\tf{alpha}\;,\;\ctrl(\tf{alpha})\;,\;\ctrl(\ctrl(\tf{pi}))\right\},\hspace{1.5em}
    M_+\defeq \cup_{n,m\in\N}\{\id{n},\swap{n,m}\}\cup M.
\end{gather*}

\begin{lemma}\label{lem:cqcreducibleplus}
    $\catcuni$ is $M_+$-reducible.
\end{lemma}
\begin{proof}
    We explain a procedure to transform any morphism of $\catcuni$ into a morphism of $\catcuni\vert_{M_+}$. Throughout the procedure, whenever we have $\ctrl(f\seqop g)$ or $\ctrl(f\parop g)$, we can split the control as explained in \cref{sec:cprop} (this is always possible because $\catcuni$ only contains endomorphisms). Then, whenever we have some $\swap{n,m}$ inside a control, we transform it into compositions of $\id{n}$, $\tf{h}$, and $\ctrl(\ctrl(\tf{pi}))$ using \labelcref{eq:symid,eq:symtop,eq:symbot,eq:cbg}. Then, using the functoriality of $\ctrl$ together with \labelcref{eq:strenght}, we can pull out all identities that are inside a control. We get a morphism in which all controls are of the form $\ctrl^n(\tf{h})$ or $\ctrl^n(\tf{alpha})$ for some $n\in\N$, and we can reduce the largest number of nested control by applying \labelcref{eq:chreduce,,eq:ccalphareduce,,eq:cccpireduce} from left to right to get a morphism of $\catcuni\vert_{M_+}$ (where \labelcref{eq:ccalphareduce} is only applied when $\alpha\ne\pi$ to avoid infinite loops). We can guarantee this process to terminate. Indeed, given a morphism $f$, let $\mu(f)\defeq\dom(\ctrl^n(x))$ be maximum such that there is a $\ctrl^n(x)$ in $f$ where $x=\tf{h}$ or $x=\tf{palpha}$ for some $\alpha\in\R$; and let $\mu_{\ne\pi}(f)\defeq \dom(\ctrl^m(x))$ be maximum such that there is a $\ctrl^m(x)$ in $f$ where $x=\tf{h}$ or $x=\tf{palpha}$ for some $\alpha\in\R\backslash\{\pi\}$. Consider the lexicographic order $(\mu(f),\mu_{\ne\pi}(f))$. 
    If $\mu(f)$ is reached by some $\ctrl^n(\tf{h})$ or $\ctrl^n(\tf{alpha})$ with $\alpha\ne\pi$, then applying \labelcref{eq:chreduce,,eq:ccalphareduce} as much as possible preserves $\mu(f)$ and decreases $\mu_{\ne\pi}(f)$; otherwise $\mu(f)$ is only reached by some $\ctrl^n(\tf{pi})$ with $n\ge3$, and applying \eqref{eq:cccpireduce} as much as possible decreases $\mu(f)$. In both cases the pair $(\mu(f),\mu_{\ne\pi}(f))$ decreases for the lexicographic order.
\end{proof}

\begin{lemma}\label{lem:cqcreducible}
    $\catcuni$ is $M$-reducible.
\end{lemma}
\begin{proof}
    By \cref{lem:cqcreducibleplus}, $\catcuni$ is $M_+$-reducible, and it only remains to prove that all identity morphisms $\id{n}$ and all swap morphisms $\swap{n,m}$ can be transformed into morphisms of $\catcuni\vert_M$. Each $\swap{n,m}$ can be transformed into compositions of $\id{n}$ and $\swap{1,1}$ using \labelcref{eq:symid,eq:symtop,eq:symbot}; each $\swap{1,1}$ can be transformed into compositions of $\id{n}$, $\tf{h}$, and $\ctrl(\ctrl(\tf{pi}))$ using \labelcref{eq:cbg}; and each $\id{n}$ can be transformed into $\ctrl(\tf{0_01})\parop\cdots\parop\ctrl(\tf{0_01})$ using \labelcref{eq:sum,eq:p0}.
\end{proof}

\begin{lemma}\label{lem:cqcconjuated}
    $\catcuni$ is a conjugated CPROP.
\end{lemma}
\begin{proof}
    $\catcuni$ is an $M$-reducible unitary CPROP (see \cref{lem:cqcunitary,lem:cqcreducible}) and \cref{lem:conjugationcondition} gives a sufficient condition \eqref{eq:conjugationcondition} for the conjugation law to hold in these settings. We thus have to prove \eqref{eq:conjugationcondition} in $\catcuni$ for any $f:n\to n$ and $g:m\to m$ in $M$, and any $k,\ell\in\N$ such that $\ell\ge1$ and $\max(n,m)\le k+\ell<n+m$.
    As $M$ contains four kinds of morphisms, there are many cases to consider. Fortunately, most cases are trivial. When $f$ or $g$ is a global phase $\tf{alpha}$, there is no such $k,\ell$. When neither $f$ nor $g$ is a Hadamard gate, $f$ and $g$ are controlled global phases (namely of the form $\ctrl(\tf{alpha})$ or $\ctrl(\ctrl(\tf{pi}))$): such morphisms, as well as their controlled versions, commute whatever the wires they are applied on (using \labelcref{eq:controlswap} and \labelcref{eq:swapconjugation} one can make the shared control wires the outermost ones, and conclude with \labelcref{eq:ctrlseq,eq:strenght} and the fact that global phases commute). Hence, by unitarity of $\catcuni$ and \labelcref{eq:swapconjugation}, both sides of \eqref{eq:conjugationcondition} are equal to $\ctrl(\swap{k,\ell}\seqop\overline{f}\seqop\swap{\ell,k})$.
    As summarized in \cref{fig:cqcconjuatedsummary}, all remaining cases are provable in $\catcuni$.
\end{proof}

\begin{figure}[t!]
    \centering
    \setlength{\tabcolsep}{5pt}
    \begin{tabularx}{\linewidth}{|c|c|c|c|c|Y|}
        \hline
        $f:n\to n$ & $g:m\to m$ & $n$ & $m$ & $(k,\ell)$ & equation \\
        \hline
        $\tf{H}$ & $\tf{H}$ & $1$ & $1$ & $(0,1)$ & \labelcref{eq:conjhh} \\
        \hline
        $\tf{H}$ & $\ctrl(\tf{alpha})$ & $1$ & $1$ & $(0,1)$ & \labelcref{eq:conjcalphah} \\
        \hline
        $\tf{H}$ & $\ctrl(\ctrl(\tf{pi}))$ & $1$ & $2$ & $(0,2)$ & \labelcref{eq:conjccpihvar} \\
        \hline
        $\tf{H}$ & $\ctrl(\ctrl(\tf{pi}))$ & $1$ & $2$ & $(1,1)$ & \labelcref{eq:conjccpih} \\
        \hline
        $\ctrl(\tf{alpha})$ & $\tf{H}$ & $1$ & $1$ & $(0,1)$ & \labelcref{eq:conjhcalpha} \\
        \hline
        $\ctrl(\ctrl(\tf{pi}))$ & $\tf{H}$ & $2$ & $1$ & $(0,2)$ & \labelcref{eq:conjhccpi,eq:controlswap,eq:swapconjugation} \\
        \hline
        $\ctrl(\ctrl(\tf{pi}))$ & $\tf{H}$ & $2$ & $1$ & $(1,1)$ & \labelcref{eq:conjhccpi,eq:controlswap,eq:swapconjugation} \\
        \hline
    \end{tabularx}
    \vspace{0.5em}
    \caption{Summary of which equation proves \eqref{eq:conjugationcondition} in \cref{lem:cqcconjuated}. \label{fig:cqcconjuatedsummary}}
\end{figure}

\subsection{Completeness}

To prove completeness of controllable quantum circuits, we use an encoding/decoding procedure between $\catuni$ and $\catcuni$. Thanks to \cref{lem:cqcreducibleplus}, we can actually simplify the proofs by conducting them using $\catcuni\vert_{M_+}$ instead of $\catcuni$. There are two key steps: the first one is to establish the relation between \labelcref{eq:zn} and the control of $\catcuni$ (see \cref{lem:ctrldecoding}); and the second is to ensure that there are enough axioms in $\axcuni$ to decode each axiom of $\axuni$ (see \cref{lem:decoding}).

\begin{definition}[encoding]\label{def:encoding}
    Let $\encoding:\catcuni\vert_{M_+}\to\catuni$ be the \emph{encoding} defined as
    \begin{gather*}
        \encoding(C_1\seqop C_2) \defeq \encoding(C_1)\seqop \encoding(C_2), \hspace{2em}
        \encoding(C_1\parop C_2) \defeq \encoding(C_1)\parop \encoding(C_2), \\[0.5em]
        \encoding(\id{n}) \defeq \id{n}, \hspace{2em}
        \encoding(\swap{n,m}) \defeq \swap{n,m}, \hspace{2em}
        \encoding(\tf{palpha}) \defeq \tf{palpha}, \\[0em]
        \encoding\left(\tf{h}\right) \defeq \tf{h}, \hspace{1.3em}
        \encoding\left(\tf{calpha_middle}\right) \defeq \tf{zalpha}, \hspace{1.3em}
        \encoding\left(\tf{ccpi_middle}\right) \defeq \tf{ccpidef}.
    \end{gather*}
\end{definition}

\begin{lemma}\label{lem:encodinginterp}
    $\intuni{\encoding(C)}=\intcuni{C}$ for any $C\in\catcuni\vert_{M_+}$.
\end{lemma}
\begin{proof}
    By induction on $C$. The base cases are easily checked, and the induction cases are proved as follows where $\diamond$ is either $\seqop$ or $\parop$.
    \begin{align*}  \intuni{\encoding(C_1\diamond C_2)}
        &=\intuni{\encoding(C_1)\diamond\encoding(C_2)}\\
        &=\intuni{\encoding(C_1)}\diamond\intuni{\encoding(C_2)}\\
        &\eqihref\intcuni{C_1}\diamond\intcuni{C_2} =\intcuni{C_1\diamond C_2}\qedhere
    \end{align*}
\end{proof}

\begin{definition}[decoding]\label{def:decoding}
    Let $\decoding:\catuni\to\catcuni\vert_{M_+}$ be the \emph{decoding} defined as
    \begin{gather*}
        \decoding(C_1\seqop C_2) \defeq \decoding(C_1)\seqop \decoding(C_2), \hspace{2em}
        \decoding(C_1\parop C_2) \defeq \decoding(C_1)\parop \decoding(C_2), \\[0.5em]
        \decoding(\id{n}) \defeq \id{n}, \hspace{2em}
        \decoding(\swap{n,m}) \defeq \swap{n,m}, \hspace{2em}
        \decoding(\tf{palpha}) \defeq \tf{palpha}, \\[0em]
        \decoding\left(\tf{h}\right) \defeq \tf{h}, \hspace{1.3em}
        \decoding\left(\tf{zalpha}\right) \defeq \tf{calpha}, \hspace{1.3em}
        \decoding\left(\tf{CNOT}\right) \defeq \tf{CNOTdef}.
    \end{gather*}
\end{definition}

\begin{lemma}\label{lem:decodingencoding}
    $\catcuni\vdash\decoding(\encoding(C))=C$ for any $C\in\catcuni\vert_{M_+} $.
\end{lemma}
\begin{proof}
    By induction on $C$. The base cases are easily proved using \labelcref{eq:hh}, and the induction cases are proved as follows where $\diamond$ is either $\seqop$ or $\parop$.
    \begin{multline*}
        \decoding(\encoding(C_1\diamond C_2))
        =\decoding(\encoding(C_1)\diamond\encoding(C_2))
        =\decoding(\encoding(C_1))\diamond\decoding(\encoding(C_2))
        \eqihref C_1\diamond C_2
    \end{multline*}
\end{proof}

\begin{lemma}\label{lem:ctrldecoding}
    $\catcuni\vdash \decoding(Z^n(\alpha))=\ctrl^{n+1}(\tf{palpha})$ for any $n\in\N$ and $\alpha\in\R$.
\end{lemma}
\begin{proof}
    By induction on $n$, where the base case $n=0$ is trivial, and where the induction case is proved as follows.
    \begin{gather*}
        \tf{ctrldecoding_var_01} \\[0.5em]
        \eqcref{def:decoding}\tf{ctrldecoding_var_02} \\[0.5em]
        \eqihref\tf{ctrldecoding_var_03}
        \eqeqref{eq:controlswap}\tf{ctrldecoding_var_04} \\[0.5em]
        \eqeqref{eq:conjugation}\tf{ctrldecoding_var_05} \\[0.5em]
        \eqeqref{eq:ctrlseq}\tf{ctrldecoding_var_06} 
        \eqeqref{eq:ccalphareduce}\tf{ctrldecoding_var_07}
        \eqeqref{}\tf{ctrldecoding_var_08}
        =\ctrl^{n+2}(\tf{palpha})
    \end{gather*}
\end{proof}

\begin{lemma}\label{lem:decoding}
    $\catuni\vdash C_1=C_2\implies\catcuni\vdash\decoding(C_1)=\decoding(C_2)$.
\end{lemma}
\begin{proof}
    We prove that $\catcuni\vdash \decoding(L)=\decoding(R)$ for any axiom $(L=R)\in\axuni$. For some axioms this is trivial, namely \labelcref{eq:p,,eq:pplus,,eq:hh,,eq:e}, because their decodings are directly axioms of $\axcuni$. Moreover, \labelcref{eq:z,eq:zplus} are decoded as the controlled versions of \labelcref{eq:p,eq:pplus} and are thus directly provable in $\catcuni$.
    The axiom \labelcref{eq:bg} is decoded as follows.
    \begin{gather*}
        \decoding\left(\tf{decoding_bg_var_01}\right)
        \eqcref{def:decoding}\tf{decoding_bg_var_02} \\[0.5em]
        \eqeqref{eq:cbg}\tf{decoding_bg_var_03}
        \eqeqref{eq:hh,eq:pplus,eq:p}\tf{decoding_bg_var_04} \\[0.5em]
        \eqeqref{eq:controlswap}\tf{decoding_bg_var_05}
        \eqcref{def:decoding}\decoding\left(\tf{decoding_bg_var_06}\right)
    \end{gather*}
    The axiom \labelcref{eq:cx} is decoded as follows.
    \begin{gather*}
        \decoding\left(\tf{decoding_zcommut_01}\right)
        \eqcref{def:decoding}\tf{decoding_zcommut_02}
        \eqeqref{eq:hh}\tf{decoding_zcommut_03} \\[0.5em]
        \eqeqref{eq:ctrlseq}\tf{decoding_zcommut_04}
        \eqeqref{eq:pplus,eq:p}\tf{decoding_zcommut_05}
        \eqeqref{eq:hh}\tf{decoding_zcommut_06}
        \eqcref{def:decoding}\decoding\left(\tf{decoding_zcommut_07}\right)
    \end{gather*}
    The axiom \labelcref{eq:cz} is decoded as follows.
    \begin{gather*}
        \decoding\left(\tf{decoding_cz_01}\right)
        \eqcref{def:decoding}\tf{decoding_cz_02}
        \eqeqref{eq:hh}\tf{decoding_cz_03} \\[0.5em]
        \eqeqref{eq:ccalphareduce}\tf{decoding_cz_04}
        \eqcref{def:decoding}\decoding\left(\tf{decoding_cz_05}\right)
    \end{gather*}
    Finally, the axiom \labelcref{eq:zn} is decoded as follows.
    \begin{gather*}
        \decoding\left(\tf{decoding_zn_var_01}\right)
        \eqcref{lem:ctrldecoding}\tf{decoding_zn_var_03}
        \eqeqref{eq:p}\tf{decoding_zn_var_04}
        \eqcref{def:decoding}\decoding\left(\tf{decoding_zn_var_04}\right)
    \end{gather*}
\end{proof}

\begin{theorem}\label{th:cqccompleteness}
    $\catcuni$ is complete.
\end{theorem}
\begin{proof}
    Let $C_1,C_2\in\catcuni(n,n)$ such that $\intcuni{C_1}=\intcuni{C_2}$. By \cref{lem:cqcreducible}, there are some $C_1',C_2'\in\catcuni\vert_{M}(n,n)$ such that $\catcuni\vdash C_i=C_i'$. \cref{lem:encodinginterp} and completeness of $\catuni$ (see \cref{th:qccompleteness}) imply $\catuni\vdash\encoding(C_1')=\encoding(C_2')$. Then, by \cref{lem:decoding}, we get $\catcuni\vdash\decoding(\encoding(C_1'))=\decoding(\encoding(C_2'))$ and \cref{lem:decodingencoding} concludes the proof as follows.
    \begin{multline*}
        \catcuni\vdash C_1
        \eqcref{lem:cqcreducible}C_1'
        \eqcref{lem:decodingencoding}\decoding(\encoding(C_1')) \\
        \eqcref{lem:decoding,th:qccompleteness}\decoding(\encoding(C_2'))
        \eqcref{lem:decodingencoding}C_2'
        \eqcref{lem:cqcreducible}C_2
    \end{multline*}
\end{proof}

It is remarkable that the axiomatization $\axcuni$ only contains axioms involving at most three qubits, each of which admits a natural and intuitive interpretation. In a sense, using the CPROP formalism makes it possible to isolate the essence of quantum circuits by abstracting away the parts of the axiomatization that are not specific to quantum computing.

Moreover, as \labelcref{eq:conjhcalpha,,eq:conjcalphah,,eq:conjhccpi,,eq:conjccpih} are instances of the conjugation law \labelcref{eq:conjugation}, one can get rid of these axioms -- yielding the axiomatization $\axccuni$ (see \cref{fig:axcuni}) containing only five axioms involving at most two qubits -- by directly considering controllable quantum circuits as a freely generated CCPROP.

\begin{corollary}\label{cor:ccqccompleteness}
    $\catccuni\defeq\cat{CCPROP}_{\sigcuni}/\axccuni$ is complete.
\end{corollary}

\section{Controllable quantum circuits with ancillas}\label{sec:acqc}

In this section, we introduce the CPROP $\catcanc$ of \emph{controllable quantum circuits with ancillas} featuring auxiliary qubits. This provides an example of a CPROP that does not consist solely of endomorphisms and demonstrates that, similarly to the vanilla quantum circuit case \cite{extensions}, ancillas can be embedded into controllable quantum circuits.

In these settings, any isometry (see \cref{ex:proplinisouni}) can be expressed. We can initialize fresh qubits in the state $\ket{0}$, and also release qubits when they become useless. To guarantee that the overall evolution is an isometry, qubits can only be released when they are in the state $\ket{0}$ to ensure that the states of other qubits are not disturbed. This condition makes the definition of $\catcanc$ a bit tedious. Contrary to $\catuni$ and $\catcuni$, the graphical language $\catcanc$ is not directly freely generated. Instead, we first define the freely generated CPROP $\catcpos$ of \emph{controllable quantum circuits with post-selection} and define $\catcanc$ to be a restriction of $\catcpos$. More precisely, $\catcpos$ is the freely generated CPROP (see \cref{def:freecprop}) whose signature $\siganc$ contains the four kinds of generators
\begin{gather*}
    \tf{palpha}:0\to 0, \hspace{2em}
    \tf{H}:1\to 1, \hspace{2em}
    \tf{winit}:0\to 1,\hspace{2em}
    \tf{wterm}:1\to 0,
\end{gather*}
where the first one is parametrized by $\alpha\in\R$. Moreover, the axiomatization $\axanc$ of controllable quantum circuits with post-selection is depicted in \cref{fig:axanc}.

\begin{definition}[controllable quantum circuits with post-selection]\label{def:cpqc}
    Let $\catcpos\defeq\cat{CPROP}_{\siganc}/\axanc$ be the graphical language of \emph{controllable quantum circuits with post-selection} whose standard interpretation $\intpos{\cdot}:\catcpos\to\catlinmap{2}$ is the identity-on-objects CPROP functor inductively defined as
    \begin{gather*}
        \intpos{\ctrl(C)}\defeq\ctrl_{\ket{1}}(\intpos{C}), \hspace{2em}
        \intpos{\tf{palpha}}\defeq\intuni{\tf{palpha}}, \hspace{2em}
        \intpos{\tf{H}}\defeq\intuni{\tf{H}}, \\[0.5em]
        \intpos{\tf{winit}}\defeq \ket{0}, \hspace{4em}
        \intpos{\tf{wterm}}\defeq \bra{0}.
    \end{gather*}
\end{definition}

Every axiom of $\axcuni$ is also in $\axanc$, and there are four new axioms: \labelcref{eq:a} allows creating a fresh qubit; \labelcref{eq:ic,,eq:icc} reveal that \tf{winit} behaves as an annihilator on control; and \labelcref{eq:ca} allows pulling out the \tf{winit} and \tf{wterm} from control.

\begin{figure}[t!]
    \begin{figurebox}\begin{center}
        \vspace{-0.5em}
        \hspace{-1em}\begin{subfigure}{3.2cm}
            \begin{equation}\label{eq:a}\tag{$\textup{a}$}
                \tf{i_left}=\tf{i_right}
            \end{equation}
        \end{subfigure}\hspace{2em}
        \begin{subfigure}{4.2cm}
            \begin{equation}\label{eq:ic}\tag{$\textup{ic}$}
                \tf{iz_var_left}=\tf{iz_right}
            \end{equation}
        \end{subfigure}\hspace{2em}
        \begin{subfigure}{4.2cm}
            \begin{equation}\label{eq:icc}\tag{$\textup{icc}$}
                \tf{icz_left}=\tf{icx_right}
            \end{equation}
        \end{subfigure}

        \vspace{0em}
        \hspace{-1em}\begin{subfigure}{2.9cm}
            \begin{equation*}\tag{$\textup{p}$}
                \tf{p_left}=\tf{p_right}
            \end{equation*}
        \end{subfigure}\hspace{2em}
        \begin{subfigure}{4.5cm}
            \begin{equation*}\tag{$\textup{p}+$}
                \tf{pplus_left}=\tf{pplus_right}
            \end{equation*}
        \end{subfigure}\hspace{2em}
        \begin{subfigure}{4.6cm}
            \begin{equation*}\tag{$\textup{hh}$}
                \tf{hh_left}=\tf{hh_right}
            \end{equation*}
        \end{subfigure}

        \vspace{1em}
        \hspace{-1em}\begin{subfigure}{7.5cm}
            \begin{equation*}\tag{$\textup{cbg}$}
                \tf{cbg_left}=\tf{cbg_right}
            \end{equation*}
        \end{subfigure}

        \vspace{0.5em}
        \hspace{-1em}\begin{subfigure}{9.2cm}
            \begin{equation*}\tag{$\textup{ce}$}
                \tf{ce_left}=\tf{ce_right}
            \end{equation*}
        \end{subfigure}

        \vspace{0.6em}\noindent\rule{\textwidth}{0.4pt}

        \vspace{0.5em}
        \hspace{-1em}\begin{subfigure}{6.6cm}
            \begin{equation}\label{eq:ca}\tag{$\textup{ca}$}
                \tf{conjugationancilla_left}=\tf{conjugationancilla_right}
            \end{equation}
        \end{subfigure}
  
        \vspace{0.5em}
        \hspace{-1em}\begin{subfigure}{6.5cm}
            \begin{equation*}\tag{$\textup{ch}_1$}
                \tf{conjhcalpha_left}=\tf{conjhcalpha_right}
            \end{equation*}
        \end{subfigure}\hspace{2em}
        \begin{subfigure}{6.5cm}
            \begin{equation*}\tag{$\textup{ch}_2$}
                \tf{conjcalphah_left}=\tf{conjcalphah_right}
            \end{equation*}
        \end{subfigure}
        
        \vspace{0.5em}
        \hspace{-1em}\begin{subfigure}{6.5cm}
            \begin{equation*}\tag{$\textup{ch}_3$}
                \tf{conjhcccpi_left}=\tf{conjhcccpi_right}
            \end{equation*}
        \end{subfigure}\hspace{2em}
        \begin{subfigure}{6.5cm}
            \begin{equation*}\tag{$\textup{ch}_4$}
                \tf{conjccpich_left}=\tf{conjccpich_right}
            \end{equation*}
        \end{subfigure}
    \end{center}\end{figurebox}
    \caption{The axioms \labelcref{eq:a,,eq:ic,,eq:icc,,eq:p,,eq:pplus,,eq:hh,,eq:cbg,,eq:ce,,eq:ca,,eq:conjhcalpha,,eq:conjcalphah,,eq:conjhccpi,,eq:conjccpih} form the complete axiomatization $\axanc$ for $\catcanc$ (where quantum circuits with ancillas are defined as a CPROP); and the subset of axioms \labelcref{eq:a,,eq:ic,,eq:icc,,eq:p,,eq:pplus,,eq:hh,,eq:cbg,,eq:ce} form the complete axiomatization $\axccanc$ for $\catccanc$ (where quantum circuits with ancillas are defined as a CCPROP).\label{fig:axanc}}
\end{figure}

Since isometries are linear maps, there is an inclusion $\catisomap{2}\hookrightarrow\catlinmap{2}$, and we can define $\catcanc$ and its standard interpretation $\intanc{\cdot}:\catcanc\to\catisomap{2}$ as follows.

\begin{definition}[controllable quantum circuits with ancillas]\label{def:qcanc}
    Let $\catcanc$ be the CPROP defined by
    \begin{gather*}
        \catcanc(n,m)\defeq\left\{C:C\in\catcpos(n,m)\;\text{s.t.}\;\intpos{C}\in\catisomap{2}(n,m)\right\}
    \end{gather*}
    for any $n,m\in\N$; and let $\intanc{\cdot}:\catcanc\to\catisomap{2}$ be the restriction of $\intpos{\cdot}$ to $\catcanc$.
\end{definition}

Notice that $\catcanc$ is indeed a CPROP because identities and swaps are interpreted in $\catisomap{2}$ and because isometries are closed under compositions and control. Moreover, $\sigcuni\subset\siganc$, which implies $\catcuni\hookrightarrow\catcanc\hookrightarrow\catcpos$. Notice also that $\catcanc(n,m)=\varnothing$ when $m<n$ because isometries have no more inputs than outputs (otherwise, there would be a loss of information).

\begin{proposition}\label{prop:aqcuniversality}
    $\catcanc$ is universal.
\end{proposition}
\begin{proof}
    For any isometry $V\in\catisomap{2}(k,n+k)$ there exists a unitary $U\in\catunimap{2}(n+k,n+k)$ such that $V=U(\ket{0}^{\parop n}\parop\idmat{2^k})$ \cite{nielsenchuang}. By universality of $\catcuni$ (see \cref{prop:cqcuniversality}), there exists $C\in\catcuni(n+k,n+k)$ such that $\intcuni{C}=U$ and $\intanc{(\tf{winit}^{\parop n}\parop\id{k})\seqop C}=V$.
\end{proof}

\begin{proposition}
    The axiomatization $\axanc$ is sound.
\end{proposition}
\begin{proof}
    Similarly to \cref{prop:qcsoundness}, it is enough to check that $\intanc{L}=\intanc{R}$ for any axiom $(L=R)\in\axanc$, which is done by direct computation.
\end{proof}

\begin{lemma}\label{lem:pulloutancilla}
    For any $C\in\catcanc(k,k+n)$ there exist $C'\in\catcuni(m+n+k,m+n+k)$ such that
    \begin{gather*}
        \catcanc\vdash\tf{ckn}=\tf{cprimemnk}.
    \end{gather*}
\end{lemma}
\begin{proof}
    We can always bend the wires to put all \tf{winit} and \tf{wterm} gates on top of the circuit using the coherence laws and \labelcref{eq:ca} (only endomorphisms can be controlled, which implies that there are always as many \tf{winit} as there are \tf{wterm} inside a control).
\end{proof}

While the generator $\tf{winit}$ corresponds to qubit initialization in the state $\ket{0}$, we can also define qubit initialization in the state $\ket{1}$ with the shorthand notations
\begin{gather*}
    \tf{not}\defeq\tf{notdef}
    \hspace{2em}\text{and}\hspace{2em}
    \tf{binit}\defeq\tf{binitdef}.
\end{gather*}

\begin{lemma}\label{lem:acqcpoints}
    \tf{winit} and \tf{binit} are points in $\catcanc$.
\end{lemma}
\begin{proof}
    We have to prove
    \begin{equation*}
        \catcanc\vdash\tf{winitcc}=\tf{winitidid}
        \hspace{2em}\text{and}\hspace{2em}
        \catcanc\vdash\tf{binitcc}=\tf{binitc}
    \end{equation*}
    for any $C\in\catcanc(n,n)$. By \cref{lem:pulloutancilla,lem:cqcreducible}, it is enough to prove these equations for the morphisms of $M$, which correspond to \labelcref{eq:ic,,eq:binitcalpha,,eq:winitch,,eq:binitch,,eq:winitccalpha,,eq:binitccalpha,,eq:winitcccpi,,eq:binitcccpi} (see \cref{fig:acqcidentities}). The proofs of these equations are given in order of dependency in \cref{app:acqcidentities}. More precisely, we have the following derivations.
    \begin{gather*}
        \tf{aqcfalse_01}
        \eqcref{lem:pulloutancilla}\tf{aqcfalse_02}
        \eqcref{lem:cqcreducible}\tf{aqcfalse_03}
        \eqeqref{eq:ca}\tf{aqcfalse_04}  \\[0.5em]
        \eqeqref{eq:ic,eq:binitcalpha,eq:winitch,eq:binitch,eq:winitccalpha,eq:binitccalpha,eq:winitcccpi,eq:binitcccpi}\tf{aqcfalse_05}
        \eqeqref{eq:a}\tf{aqcfalse_06}
    \end{gather*}
    \begin{gather*}
        \tf{aqctrue_01}
        \eqcref{lem:pulloutancilla}\tf{aqctrue_02}
        \eqcref{lem:cqcreducible}\tf{aqctrue_03}
        \eqeqref{eq:ca}\tf{aqctrue_04} \\[0.5em]
        \eqeqref{eq:ic,eq:binitcalpha,eq:winitch,eq:binitch,eq:winitccalpha,eq:binitccalpha,eq:winitcccpi,eq:binitcccpi}\tf{aqctrue_05}
        \eqcref{lem:cqcreducible}\tf{aqctrue_06}
        \eqcref{lem:pulloutancilla}\tf{aqctrue_07}\qedhere
    \end{gather*}
\end{proof}

\begin{figure}[t!]
    \begin{figurebox}\begin{center}
        \vspace{-0.5em}
        \hspace{-1em}\begin{subfigure}{4.7cm}
            \begin{equation}\label{eq:notnot}
                \tf{notnot_01}=\tf{notnot_03}
            \end{equation}
        \end{subfigure}\hspace{1.5em}
        \begin{subfigure}{4.2cm}
            \begin{equation}\label{eq:binitcalpha}
                \tf{binitcalpha_01}=\tf{binitcalpha_07}
            \end{equation}
        \end{subfigure}\hspace{1.5em}
        \begin{subfigure}{4.2cm}
            \begin{equation}\label{eq:binitccpi}
                \tf{binitccpi_01}=\tf{binitccpi_10}
            \end{equation}
        \end{subfigure}
        
        \vspace{1em}
        \hspace{-1em}\begin{subfigure}{4.2cm}
            \begin{equation}\label{eq:winitch}
                \tf{winitch_01}=\tf{winitch_04}
            \end{equation}
        \end{subfigure}\hspace{2em}
        \begin{subfigure}{4.2cm}
            \begin{equation}\label{eq:binitch}
                \tf{binitch_01}=\tf{binitch_04}
            \end{equation}
        \end{subfigure}

        \vspace{1em}
        \hspace{-1em}\begin{subfigure}{4.2cm}
            \begin{equation}\label{eq:winitccalpha}
                \tf{winitccalpha_01}=\tf{winitccalpha_04}
            \end{equation}
        \end{subfigure}\hspace{2em}
        \begin{subfigure}{4.2cm}
            \begin{equation}\label{eq:binitccalpha}
                \tf{binitccalpha_01}=\tf{binitccalpha_06}
            \end{equation}
        \end{subfigure}
        
        \vspace{1em}
        \hspace{-1em}\begin{subfigure}{4.2cm}
            \begin{equation}\label{eq:winitcccpi}
                \tf{winitcccpi_01}=\tf{winitcccpi_04}
            \end{equation}
        \end{subfigure}\hspace{2em}
        \begin{subfigure}{4.2cm}
            \begin{equation}\label{eq:binitcccpi}
                \tf{binitcccpi_01}=\tf{binitcccpi_04}
            \end{equation}
        \end{subfigure}
    \end{center}\end{figurebox}
    \caption{Some equations provable in $\catcanc$. \label{fig:acqcidentities}}
\end{figure}

\begin{theorem}\label{th:acqccompleteness}
    $\catcanc$ is complete.
\end{theorem}
\begin{proof}
    A similar proof has been conducted in the context of PROP in \cite{extensions}. We proceed exactly the same way, using a variant of the cosine-sine decomposition for which we give an improved version of the proof in \cref{app:gcsd}. The detailed proof is given in \cref{app:acqccompleteness}.
\end{proof}

Interestingly, while the necessity of axioms acting on an unbounded number of qubits (see \cref{sec:qc}) can be circumvented in $\catcuni$, this does not seem to be the case in $\catcanc$ where \labelcref{eq:ca} involves an arbitrary number of qubits. We believe that we cannot directly derive \labelcref{eq:ca} using the other axioms of $\axanc$, but of course, it may be possible to replace it with more local axioms. Nonetheless, as mentioned previously, \labelcref{eq:conjhcalpha,,eq:conjcalphah,,eq:conjhccpi,,eq:conjccpih} are instances of the conjugation law \labelcref{eq:conjugation}. This is also the case of \labelcref{eq:ca}. Then, one can get rid of these axioms -- yielding the axiomatization $\axccanc$ (see \cref{fig:axanc}) containing only eight axioms involving at most two qubits -- by directly considering controllable quantum circuits with ancillas as a freely generated CCPROP.

\begin{corollary}\label{cor:accqccompleteness}
    The restriction $\catccanc$ of $\catccpos\defeq\cat{CCPROP}_{\siganc}/\axccanc$ to isometries, defined by
    \begin{gather*}
        \catccanc(n,m)\defeq\left\{C:C\in\catccpos(n,m)\;\text{s.t.}\;\intpos{C}\in\catisomap{2}(n,m)\right\}
    \end{gather*}
    for any $n,m\in\N$, is complete.
\end{corollary}

\section{Polycontrolled PROPs}\label{sec:poly}

So far, we focused on PROPs that are equipped with a single control functor. It is natural to consider PROPs that admit multiple coexisting control functors.

\begin{definition}[polycontrolled PROP]
    A \emph{polycontrolled PROP (poly-CPROP)} is a PROP equipped with multiple control functors.
\end{definition}

As illustrated by the following example, considering poly-CPROPs instead of CPROPs can sometimes simplify diagrammatic reasoning. 

\begin{example}\label{ex:qcjustphase}
    While controllable quantum circuits are defined using two generators and one control functor (see \cref{def:cqc}), we can also define a similar graphical language using one generator and two control functor similarly to the language introduced in \cite{qcarejustphase}. If we only have global phases \tf{palpha} for any $\alpha\in\R$, and two control functors $\ctrl_Z$ and $\ctrl_X$, we can still express any unitaries in $\catunimap{2}$ with the interpretation defined as
    \begin{gather*}
        \interp{\tf{palpha}}\defeq e^{\ii\alpha}, \hspace{2em}
        \interp{\ctrl_Z(C)}\defeq \ctrl_{\ket{1}}(\interp{C}), \hspace{2em}
        \interp{\ctrl_X(C)}\defeq \ctrl_{\ket{-}}(\interp{C}).
    \end{gather*}
    Intuitively $\ctrl_Z$ allows implementing Z-rotations while $\ctrl_X$ allows implementing X-rotations. For instance, the Hadamard gate can be implemented by 
    \begin{gather*}
        \tf{minuspiover4}\parop\ctrl_Z(\tf{piover2})\seqop\ctrl_X(\tf{piover2})\seqop\ctrl_Z(\tf{piover2}).
    \end{gather*}
\end{example}

In the following, we review some notable properties that multiple coexisting control functors may satisfy.

\begin{definition}[compatibility]\label{def:compatibility}
    Two control functors $\ctrl_1,\ctrl_2:\cat{P}_\textup{endo}\to\cat{P}_\textup{endo}$ are said to be \emph{compatible} if \labelcref{eq:compatibility} is satisfied for any $f\in\cat{P}(n,n)$.
    \begin{equation}\label{eq:compatibility}
        \begin{array}{ccc}
            (\swap{1,1}\parop \id{n})\seqop\ctrl_1(\ctrl_2 (f))&=&\ctrl_2(\ctrl_1 (f))\seqop(\swap{1,1}\parop \id{n}) \\[0.3em]
            \tf{compatibility_left}&=&\tf{compatibility_right}
        \end{array}
    \end{equation}
\end{definition}

\begin{example}
    In $\catlinmap{2}$, $\ctrl_{\ket{1}}$ is compatible with $\ctrl_{\ket{0}}$ because each control can be obtained by conjugating the control wire of the other control with a $X$ gate that slides through the swaps. However, given a self-inverse linear map $A:1\to1\notin\{\idmat{2},-\idmat{2},Z,-Z\}$ (meaning $A\seqop A=\id{1}$), $\ctrl_{\ket{1}}$ is not compatible with the control functor $\ctrl_A:U\mapsto \idmat{2}\parop A_nUA_n$ where $U$ acts on $n$ qubits and where $A_n$ applies $A$ gates to $n$ qubits in parallel. Intuitively, $\ctrl_A$ is independent of the state of the control qubit and always conjugates the target by a series of $A$ gates.
\end{example}

\begin{definition}[commutativity]\label{def:commutativity}
    Two control functors $\ctrl_1,\ctrl_2:\cat{P}_\textup{endo}\to\cat{P}_\textup{endo}$ are said to \emph{commute} if \labelcref{eq:commutativity} is satisfied for any $f,g\in\cat{P}(n,n)$.
    \begin{equation}\label{eq:commutativity}
        \begin{array}{ccc}
            \ctrl_1(f)\seqop\ctrl_2(g) &=& \ctrl_2(g)\seqop \ctrl_1(f)\\[0.3em]
            \tf{commutativity_left}&=&\tf{commutativity_right}
        \end{array}
    \end{equation}
\end{definition}

\begin{example}
    In $\catlinmap{2}$, $\ctrl_{\ket{1}}$ commutes with $\ctrl_{\ket{0}}$ but not with $\ctrl_{\ket{-}}$.
\end{example}

\begin{definition}[conversion]\label{def:controlconjugacy}
    Given two control functors $\ctrl_1,\ctrl_2:\cat{P}_\textup{endo}\to\cat{P}_\textup{endo}$, we say that a morphism $g\in\cat{P}(1,1)$ \emph{converts} $\ctrl_1$ into $\ctrl_2$ if \labelcref{eq:controlconjugacy} is satisfied for any $f\in\cat{P}(n,n)$.
    \begin{equation}\label{eq:controlconjugacy}
        \begin{array}{ccc}
            (g\parop\id{n})\seqop\ctrl_1(f) &=& \ctrl_2(f)\seqop(g\parop\id{n}) \\[0.3em]
            \tf{controlconjugatedness_left}&=&\tf{controlconjugatedness_right}
        \end{array}
    \end{equation}
\end{definition}

\begin{example}
    In $\catlinmap{2}$, $X$ converts $\ctrl_{\ket{1}}$ into $\ctrl_{\ket{0}}$ (and vice versa). More generally, in $\catlinmap{d}$, $X_{+1}\defeq\sum_k\ketbra{k+1\bmod d}{k}$ converts $\ctrl_{\ket{k}}$ into $\ctrl_{\ket{k+1\bmod d}}$.
\end{example}

\begin{definition}[exhaustivity]\label{def:exhaustivity}
    A family of $\ell$ control functors $\ctrl_k:\cat{P}_\textup{endo}\to\cat{P}_\textup{endo}$ (where $k\in\natset{\ell}$) is said to be \emph{exhaustive} if \labelcref{eq:exhaustivity} is satisfied for any $f\in\cat{P}(n,n)$.
    \begin{equation}\label{eq:exhaustivity}
        \begin{array}{ccc}
            \ctrl_1(f)\seqop \ldots \seqop \ctrl_\ell(f) &=& \id{1}\parop f\\[0.3em]
            \tf{exhaustivity_left}&=&\tf{exhaustivity_right}
        \end{array}
    \end{equation}
\end{definition}

\begin{example}\label{ex:exhaustivity}
    In $\catlinmap{d}$, the family of control functors $(\ctrl_{\ket{k}})_{0\le k<d}$ is exhaustive. Intuitively, each $\ctrl_{\ket{k}}(f)$ fires $f$ if the control qudit is in the state $\ket{k}$, which implies that if we apply all these controlled gates sequentially one after the other, then $f$ is fired exactly once.
\end{example}

\begin{example}\label{ex:onerig}
    The control structure presented in \cite{onerig} is captured as a special case of our formalism. More precisely, the authors consider a PROP $\cat{P}$ containing only endomorphisms (meaning $\cat{P}=\cat{P}_\textup{endo}$) including a specific involution $x:1\to1$ (meaning $x\seqop x=\id{1}$), in which there are two specific functors satisfying eight equations. In our settings, $\cat{P}$ exactly corresponds to a polycontrolled PROP with two commuting exhaustive control functors that are converted one into the other using the morphism $x$.
\end{example}

\section{Controlled SMCs}\label{sec:smc}

So far, we focused on SMCs that are PROPs. Nevertheless, the notion of control functor and its diagrammatic representation extend straightforwardly to arbitrary SMCs. The main difference is that a control functor is specific to a certain object.

\begin{definition}[control structure]\label{def:controlfunctorsmc}
    Given an SMC $\cat{C}$, a \emph{control structure} is a family of control functors $\ctrl^{A}:\cat{C}_\textup{endo}\to\cat{C}_\textup{endo}$ parametrized by objects $A\in\obj(\cat{C})$, mapping any object $B\in\obj(\cat{C})$ to the object $A\parop B$ and any morphism $f:B\to B$ to the morphism $\ctrl^A (f): A\parop B\to A\parop B$ while satisfying the additional coherence laws \labelcref{eq:controlidsmc,eq:controltensorsmc,eq:strenghtsmc,eq:controlswapsmc,eq:swapconjugationsmc}.
\end{definition}

\begin{figure}[t!]
    \begin{figurebox}\begin{center}
        \vspace{-0.5em}
        \hspace{-1em}\begin{subfigure}{5.3cm}
            \begin{equation}\label{eq:ctrlidsmc}
                \begin{array}{ccc}
                    \ctrl^A(\id{B})&=&\id{A\parop B} \\[0.3em]
                    \tf{ctrlid_smc_left}&=&\tf{ctrlid_smc_right}
                \end{array}
            \end{equation}
        \end{subfigure}\hspace{2em}
        \begin{subfigure}{7.4cm}
            \begin{equation}\label{eq:ctrlseqsmc}
                \begin{array}{ccc}
                    \ctrl^A(f\seqop g)&=&\ctrl^A(f)\seqop\ctrl^A(g) \\[0.3em]
                    \tf{ctrlseq_smc_left}&=&\tf{ctrlseq_smc_right}
                \end{array}
            \end{equation}
        \end{subfigure}

        \vspace{1em}
        \hspace{-1em}\begin{subfigure}{4.9cm}
            \begin{equation}\label{eq:controlidsmc}
                \begin{array}{ccc}
                    \ctrl^I(f)&=&f \\[0.3em]
                    \tf{controlid_smc_left}&=&\tf{controlid_smc_right}
                \end{array}
            \end{equation}
        \end{subfigure}\hspace{1em}
        \begin{subfigure}{6.8cm}
            \begin{equation}\label{eq:controltensorsmc}
                \begin{array}{ccc}
                    \ctrl^{A\parop B}(f)&=&\ctrl^A(\ctrl^B (f)) \\[0.3em]
                    \tf{controltensor_smc_left}&=&\tf{controltensor_smc_right}
                \end{array}
            \end{equation}
        \end{subfigure}

        \vspace{1em}
        \hspace{-1em}\begin{subfigure}{7.2cm}
            \begin{equation}\label{eq:strenghtsmc}
                \begin{array}{ccc}
                    \ctrl^A(f\parop \id{B})&=&\ctrl^A(f)\parop \id{B} \\[0.3em]
                    \tf{strenght_smc_left}&=&\tf{strenght_smc_right}
                \end{array}
            \end{equation}
        \end{subfigure}

        \vspace{1em}
        \hspace{-1em}\begin{subfigure}{11.1cm}
            \begin{equation}\label{eq:controlswapsmc}
                \begin{array}{ccc}
                    (\swap{B,A}\parop \id{C})\seqop\ctrl^{A\parop B}(f)&=&\ctrl^{B\parop A}(f)\seqop(\swap{B,A}\parop \id{C}) \\[0.3em]
                    \tf{ctrlswap_smc_left}&=&\tf{ctrlswap_smc_right}
                \end{array}
            \end{equation}
        \end{subfigure}

        \vspace{1em}
        \hspace{-1em}\begin{subfigure}{13.1cm}
            \begin{equation}\label{eq:swapconjugationsmc}
                \begin{array}{ccc}
                    \ctrl^A(\tau_{B,C,D,E}\seqop f\seqop\tau_{B,D,C,E}) &=& \tau_{A\parop B,C,D,E}\seqop \ctrl^A(f)\seqop\tau_{A\parop B,D,C,E} \\[0.3em]
                    \tf{swapconjugation_smc_right}&=&\tf{swapconjugation_smc_left}
                \end{array}
            \end{equation}
        \end{subfigure}
    \end{center}\end{figurebox}
    \caption{Coherence laws of control functors where $\tau_{A,B,C,D}\defeq(\id{A}\parop \swap{B,C}\parop \id{D})$. \label{fig:coherencelawscsmc}}
\end{figure}

Notice that \labelcref{eq:ctrlidsmc,eq:ctrlseqsmc} hold by definition of the control as a functor; \labelcref{eq:controlidsmc,eq:controltensorsmc} guarantee that the control functor preserves the monoidal structure on objects; and \labelcref{eq:strenghtsmc,,eq:controlswapsmc,,eq:swapconjugationsmc} respectively correspond to \labelcref{eq:strenght,,eq:controlswap,,eq:swapconjugation}. Moreover, instantiating this generalized definition in the context of PROPs, we recover \cref{def:controlfunctor}.

The notion of points can also be extended to arbitrary SMCs. Moreover, if $\ctrl^A$ and $\ctrl^B$ have points (respectively denoted $\tf{winit}_A,\tf{binit}_A$ and $\tf{winit}_B,\tf{binit}_B$), then $\tf{winit}_A\parop\tf{ginit}_B$ and $\tf{binit}_A\parop\tf{binit}_B$ are points for $\ctrl^{A\parop B}$ for any $\tf{ginit}_B:I\to B$, as well as $\tf{ginit}_A\parop\tf{winit}_B$ and $\tf{binit}_A\parop\tf{binit}_B$ for any $\tf{ginit}_A:I\to A$.

\section*{Conclusion}\label{sec:conc}

In this paper, we introduced controlled PROPs as an extension of the standard PROP formalism enabling diagrammatic reasoning with control as a constructor. Beyond quantum computing, we expect controlled PROPs to be useful in other domains where conditional behavior plays a central role, such as classical reversible computing. Note that the notion of control functor can be straightforwardly generalized to arbitrary symmetric monoidal categories --~as explained in \cref{sec:smc}~-- and also to weaker structures like braided monoidal categories. 

We considered, in this paper, control as a constructor, independently of its actual physical implementation. We leave a full development of this question for future investigation; however, we mention a no-go theorem concerning the possibility of controlling unknown operations using the standard quantum circuit formalism \cite{araujo}. As pointed out by the authors of this result, the no-go theorem does not apply in practice, as any physical implementation of an unknown unitary provides additional information that makes control possible. Moreover, there is evidence that such control can be implemented in practice \cite{friis,zhou}. The authors of the no-go theorem even argue that the quantum circuit formalism should be extended to capture the possibility of controlling unknown unitaries, which is a feature of the graphical language of controllable quantum circuits we introduced (see \cref{sec:cqc}). Furthermore, in the present paper, this no-go theorem does not apply, as we work in a white-box setting, i.e. we assume that the implementation of the unitary as a circuit is known.

Finally, another direction for future research is to consider controllable quantum circuits for quantum circuit optimization tasks, since they provide simpler rewriting rules than those of vanilla quantum circuits.

%% file: appendices.tex

\section{Proof of the equations of \cref{fig:cqcidentities}}\label[appendix]{app:cqcidentities}

\begin{proof}[Proof of \cref{eq:p0} in $\catcuni$]
    \begin{gather*}
        \tf{0_01}
        \eqeqref{eq:p}\tf{0_02}
        \eqeqref{eq:pplus}\tf{0_03}
        \eqeqref{eq:p}\tf{0_04}
    \end{gather*}
\end{proof}

\begin{proof}[Proof of \cref{eq:cph} in $\catcuni$]
    \begin{gather*}
        \tf{heuler_01}
        \eqeqref{eq:hh,eq:p0}\tf{heuler_02}
        \eqeqref{eq:ce}\tf{heuler_03}
    \end{gather*}
\end{proof}

\begin{proof}[Proof of \cref{eq:calphaminus} in $\catcuni$]
    \begin{gather*}
        \tf{xcalpha_01}
        \eqeqref{eq:hh}\tf{xcalpha_02} \\[0.5em]
        \eqeqref{eq:ce}\tf{xcalpha_03} \\[0.5em]
        \eqeqref{eq:p,eq:pplus}\tf{xcalpha_04} \\[0.5em]
        \eqeqref{eq:pplus}\tf{xcalpha_05}
        \eqeqref{eq:cph}\tf{xcalpha_06}
    \end{gather*}
\end{proof}

\begin{proof}[Proof of \cref{eq:cphvar} in $\catcuni$]
    \begin{gather*}
        \tf{heulervar_01}
        \eqeqref{eq:cph}\tf{heulervar_02}
        \eqeqref{eq:pplus}\tf{heulervar_03} \\[0.5em]
        \eqeqref{eq:hh}\tf{heulervar_04}
        \eqeqref{eq:calphaminus}\tf{heulervar_05}
    \end{gather*}
\end{proof}

\begin{proof}[Proof of \cref{eq:chreduce} in $\catcuni$]
    \begin{gather*}
        \tf{chaltdef_01}
        \eqeqref{eq:cphvar}\tf{chaltdef_02} \\[0.5em]
        \eqeqref{eq:pplus}\tf{chaltdef_03} \\[0.5em]
        \eqeqref{eq:conjcalphah}\tf{chaltdef_04} \\[0.5em]
        \eqeqref{eq:pplus,eq:p0}\tf{chaltdef_05}
        \eqeqref{eq:conjhcalpha}\tf{chaltdef_06} \\[0.5em]
        \eqeqref{eq:conjhcalpha}\tf{chaltdef_07}
        \eqeqref{eq:pplus}\tf{chaltdef_08} \\[0.5em]
        \eqeqref{eq:conjcalphah}\tf{chaltdef_09} \\[0.5em]
        \eqeqref{eq:pplus,eq:p0}\tf{chaltdef_10}
        \eqeqref{eq:conjhcalpha}\tf{chaltdef_11}
    \end{gather*}
\end{proof}

\begin{proof}[Proof of \cref{eq:conjhh} in $\catcuni$]
    \begin{gather*}
        \tf{conjhh_01}
        \eqeqref{eq:hh}\tf{conjhh_02} 
        \eqeqref{eq:chreduce}\tf{conjhh_03} \\[0.5em]
        \eqeqref{eq:p,eq:pplus}\tf{conjhh_04} \\[0.5em]
        \eqeqref{eq:hh,eq:p0,eq:pplus}\tf{conjhh_05} \\[0.5em]
        \eqeqref{eq:cphvar}\tf{conjhh_06} \\[0.5em]
        \eqeqref{eq:hh,eq:p0,eq:pplus}\tf{conjhh_07} \\[0.5em]
        \eqeqref{eq:cphvar}\tf{conjhh_08} 
        \eqeqref{eq:chreduce}\tf{conjhh_09}
    \end{gather*}
\end{proof}

\begin{proof}[Proof of \cref{eq:ccalphareduce} in $\catcuni$]
    \begin{gather*}
        \tf{ccalphadef_01}
        \eqeqref{eq:pplus}\tf{ccalphadef_02}
        \eqeqref{eq:hh,eq:p,eq:pplus}\tf{ccalphadef_03} \\[0.5em]
        \eqeqref{eq:calphaminus}\tf{ccalphadef_04}
        \eqeqref{eq:ctrlseq}\tf{ccalphadef_05} \\[0.5em]
        \eqeqref{eq:ctrlseq}\tf{ccalphadef_06} \\[0.5em]
        \eqeqref{eq:pplus}\tf{ccalphadef_07} \\[0.5em]
        \eqeqref{eq:conjcalphah}\tf{ccalphadef_08} \\[0.5em]
        \eqeqref{eq:pplus,eq:p0}\tf{ccalphadef_09}
        \eqeqref{eq:conjhcalpha}\tf{ccalphadef_10}
    \end{gather*}
\end{proof}

\begin{proof}[Proof of \cref{eq:conjhccpivar} in $\catcuni$]
    \begin{gather*}
        \tf{conjhccpi_04}
        \eqeqref{eq:controlswap}\tf{conjhccpi_03}
        \eqeqref{eq:conjhccpi}\tf{conjhccpi_02}
        \eqeqref{eq:controlswap}\tf{conjhccpi_01}
    \end{gather*}
\end{proof}

\begin{proof}[Proof of \cref{eq:conjccpihvar} in $\catcuni$]
    \begin{gather*}
        \tf{conjccpih_04}
        \eqeqref{eq:controlswap}\tf{conjccpih_03}
        \eqeqref{eq:conjccpih}\tf{conjccpih_02}
        \eqeqref{eq:controlswap}\tf{conjccpih_01}
    \end{gather*}
\end{proof}

\begin{proof}[Proof of \cref{eq:cccpireduce} in $\catcuni$]
    \begin{gather*}
        \tf{cccpidef_01}
        \eqeqref{}\tf{cccpidef_02}
        \eqeqref{eq:ccalphareduce}\tf{cccpidef_03} \\[0.5em]
        \eqeqref{}\tf{cccpidef_04}
        \eqeqref{eq:conjhccpivar}\tf{cccpidef_05} \\[0.5em]
        \eqeqref{eq:conjhcalpha}\tf{cccpidef_06} \\[0.5em]
        \eqeqref{eq:p,eq:pplus}\tf{cccpidef_07} \\[0.5em]
        \eqeqref{eq:conjccpihvar}\tf{cccpidef_08} \\[0.5em]
        \eqeqref{eq:pplus,eq:p}\tf{cccpidef_09}
        \eqeqref{eq:conjhcalpha}\tf{cccpidef_10}
    \end{gather*}
\end{proof}

\section{Proof of the equations of \cref{fig:acqcidentities}}\label[appendix]{app:acqcidentities}

\begin{proof}[Proof of \cref{eq:notnot} in $\catcanc$]
    \begin{gather*}
        \tf{notnot_01}
        =\tf{notnot_02}
        \eqeqref{eq:hh,eq:pplus,eq:p}\tf{notnot_03}
    \end{gather*}
\end{proof}

\begin{proof}[Proof of \cref{eq:binitcalpha} in $\catcanc$]
    \begin{gather*}
        \tf{binitcalpha_01}
        =\tf{binitcalpha_03}
        \eqeqref{eq:calphaminus}\tf{binitcalpha_04}
        \eqeqref{eq:ic}\tf{binitcalpha_05}
        =\tf{binitcalpha_07}
    \end{gather*}
\end{proof}

\begin{proof}[Proof of \cref{eq:binitccpi} in $\catcanc$]
    \begin{gather*}
        \tf{binitccpi_01}
        =\tf{binitccpi_02}
        \eqeqref{eq:hh,eq:p,eq:pplus,eq:notnot}\tf{binitccpi_03} \\[0.5em]
        =\tf{binitccpi_04} \\[0.5em]
        =\tf{binitccpi_05} \\[0.5em]
        \eqeqref{eq:ccalphareduce}\tf{binitccpi_06}
        \eqeqref{eq:p}\tf{binitccpi_07}
        \eqeqref{eq:hh}\tf{binitccpi_08}
        \eqeqref{eq:icc}\tf{binitccpi_09}
        =\tf{binitccpi_10}
    \end{gather*}
\end{proof}

\begin{proof}[Proof of \cref{eq:winitch} in $\catcanc$]
    \begin{gather*}
        \tf{winitch_01}
        \eqeqref{eq:chreduce}\tf{winitch_02}
        \eqeqref{eq:icc}\tf{winitch_03} \\[0.5em]
        \eqeqref{eq:hh,eq:pplus,eq:p0}\tf{winitch_04}
    \end{gather*}
\end{proof}

\begin{proof}[Proof of \cref{eq:binitch} in $\catcanc$]
    \begin{gather*}
        \tf{binitch_01}
        \eqeqref{eq:chreduce}\tf{binitch_02}
        \eqeqref{eq:binitccpi}\tf{binitch_03} \\[0.5em]
        \eqeqref{eq:cphvar}\tf{binitch_04}
    \end{gather*}
\end{proof}

\begin{proof}[Proof of \cref{eq:winitccalpha} in $\catcanc$]
    \begin{gather*}
        \tf{winitccalpha_01}
        \eqeqref{eq:ccalphareduce}\tf{winitccalpha_02}
        \eqeqref{eq:ic,eq:icc}\tf{winitccalpha_03} \\[0.5em]
        \eqeqref{eq:hh,eq:pplus,eq:p0}\tf{winitccalpha_04}
    \end{gather*}
\end{proof}

\begin{proof}[Proof of \cref{eq:binitccalpha} in $\catcanc$]
    \begin{gather*}
        \tf{binitccalpha_01}
        \eqeqref{eq:ccalphareduce}\tf{binitccalpha_02} \\[0.5em]
        \eqeqref{eq:binitccpi}\tf{binitccalpha_03}
        \eqeqref{eq:calphaminus}\tf{binitccalpha_04} \\[0.5em]
        \eqeqref{eq:pplus,eq:p0}\tf{binitccalpha_05}
        \eqeqref{eq:hh,eq:pplus,eq:p}\tf{binitccalpha_06}
    \end{gather*}
\end{proof}

\begin{proof}[Proof of \cref{eq:winitcccpi} in $\catcanc$]
    \begin{gather*}
        \tf{winitcccpi_01}
        \eqeqref{eq:cccpireduce}\tf{winitcccpi_02}
        \eqeqref{eq:winitccalpha}\tf{winitcccpi_03} \\[0.5em]
        \eqeqref{eq:hh,eq:pplus,eq:p}\tf{winitcccpi_04}
    \end{gather*}
\end{proof}

\begin{proof}[Proof of \cref{eq:binitcccpi} in $\catcanc$]
    \begin{gather*}
        \tf{binitcccpi_01}
        \eqeqref{eq:cccpireduce}\tf{binitcccpi_02}
        \eqeqref{eq:binitccalpha}\tf{binitcccpi_03} \\[0.5em]
        \eqeqref{eq:ccalphareduce}\tf{binitcccpi_04}
    \end{gather*}
\end{proof}

\section{Generalized cosine-sine decomposition}\label[appendix]{app:gcsd}

We provide a self-contained proof of the \emph{generalized cosine-sine decomposition (GCSD)} (introduced for a similar purpose in \cite{extensions}), which is a generalization of the \emph{cosine-sine decomposition (CSD)} \cite{csdhistory}. All considered matrices are complex, and we write $A:n\times m$ when the matrix $A$ has $n$ rows and $m$ columns. 

First, recall some usual decompositions of an arbitrary square matrix $A:n\times n$. The \emph{RQ-decomposition (RQD)} states that there exist a unitary matrix $Q:n\times n$ and an upper-triangular matrix $R:n\times n$ with real non-negative diagonal elements such that $A=RQ$. The \emph{QL-decomposition (QLD)} states that there exist a unitary matrix $Q:n\times n$ and a lower-triangular matrix $L:n\times n$ with real non-negative diagonal elements such that $A=QL$. The \emph{singular value decomposition (SVD)} states that there exist two unitary matrices $U,V:n\times n$ and a real diagonal matrix $D=\diag(d_1,d_2,\dots,d_n)$ satisfying $d_1\ge d_2\ge\dots\ge d_n\ge 0$ and such that $A=UDV$ \cite{matrixalgebra}.

\begin{lemma}\label{lem:triangulardiagonal}
    If $L:n\times n$ is a lower-triangular matrix and $L^\dagger L$ is diagonal with non-zero diagonal elements, then $L$ is diagonal. If $R:n\times n$ is an upper-triangular matrix and $RR^\dagger$ is diagonal with non-zero diagonal elements, then $R$ is diagonal.
\end{lemma}
\begin{proof}
    We prove the statement for a lower-triangular matrix $L:n\times n$ by induction on $n$ (the statement for an upper-triangular matrix $R$ follows because $R^\dagger$ is a lower-triangular matrix such that $(R^\dagger)^\dagger R^\dagger=RR^\dagger$ is diagonal with non-zero elements). The base case $n=1$ is trivial. For the induction case where $L:(n+1)\times(n+1)$, let $D\defeq L^\dagger L$. Notice that $D_{n+1,n+1}=|L_{n+1,n+1}|^2\ne 0$ and $D_{i,n+1}=\overline{L_{n+1,i}}L_{n+1,n+1}=0$ for any $i\in\natset{n}$, which implies $L_{n+1,n+1}\ne 0$ and $L_{n+1,i}=0$. Hence, all non-diagonal elements of the last row and last column of $L$ are zero, and we can use the IH on the submatrix obtained by deleting the last row and the last column of $L$.
\end{proof}

\begin{lemma}[\cite{csdhistory}]\label{lem:csd}
    Given a unitary matrix $U:2n\times 2n$, the \emph{cosine-sine decomposition (CSD)} states that there exist unitary matrices $A_0,A_1,B_0,B_1:n\times n$, and real diagonal matrices $D_c,D_s:n\times n$ with non-negative elements such that $D_c^2+D_s^2=\idmat{n}$ and 
    \begin{gather*}
        U=\left[\begin{array}{c|c}
            A_0 & 0 \\
            \hline
            0 & A_1
        \end{array}\right]\cdot\left[\begin{array}{c|c}
            D_c & -D_s\\
            \hline
            D_s & D_c
        \end{array}\right]\cdot\left[\begin{array}{c|c}
            B_0 & 0 \\
            \hline
            0 & B_1
        \end{array}\right].
    \end{gather*}
\end{lemma}

\begin{lemma}\label{lem:gcsd}
    Given a unitary matrix $U:(2n-k)\times (2n-k)$ where $n\ge k\ge0$, the \emph{generalized cosine-sine decomposition (GCSD)} states that there exist unitary matrices $A_0,B_0:(n-k)\times (n-k)$, $A_1,B_1:n\times n$ and real diagonal matrices $D_c,D_s:(n-k)\times (n-k)$ with non-negative elements such that $D_c^2+D_s^2=\idmat{n-k}$ and
    \begin{gather*}
        \left[\begin{array}{c|c}
            \idmat{k} & 0 \\
            \hline
            0 & U
        \end{array}\right]=\left[\begin{array}{c|c|c}
            \idmat{k} & 0 & 0 \\
            \hline
            0 & A_0 & 0 \\
            \hline
            0 & 0 & A_1
        \end{array}\right]\cdot\left[\begin{array}{c|c|c|c}
            \idmat{k} & 0 & 0 & 0 \\
            \hline
            0 & D_c & 0 & -D_s\\
            \hline
            0 & 0 & \idmat{k} & 0\\
            \hline
            0 & D_s & 0 & D_c
        \end{array}\right]\cdot\left[\begin{array}{c|c|c}
            \idmat{k} & 0 & 0 \\
            \hline
            0 & B_0 & 0 \\
            \hline
            0 & 0 & B_1
        \end{array}\right].
    \end{gather*}
\end{lemma}
\begin{proof}
    Let $\hat{U}\defeq \idmat{k}\oplus U$. We can split $U$ into four blocks $U_{00}:(n-k)\times (n-k)$, $U_{01}:(n-k)\times n$, $U_{10}:n\times (n-k)$ and $U_{11}:n\times n$. Notice that the unitarity of $U$ implies the unitarity of $\hat{U}$.
    \begin{gather*}
        \hat{U}=\left[\begin{array}{c|c|c}
            \idmat{k} & 0 & 0 \\
            \hline
            0 & U_{00} & U_{01} \\
            \hline
            0 & U_{10} & U_{11}
        \end{array}\right]
    \end{gather*}
    Then, let $A_0DB_0$ be the SVD of $U_{00}$, which yields the following decomposition.
    \begin{gather*}
        \hat{U}=\left[\begin{array}{c|c|c}
            \idmat{k} & 0 & 0 \\
            \hline
            0 & A_0DB_0 & U_{01} \\
            \hline
            0 & U_{10} & U_{11}
        \end{array}\right]
        =\left[\begin{array}{c|c|c}
            \idmat{k} & 0 & 0 \\
            \hline
            0 & A_0 & 0 \\
            \hline
            0 & 0 & \idmat{n}
        \end{array}\right]\cdot\left[\begin{array}{c|c|c}
            \idmat{k} & 0 & 0 \\
            \hline
            0 & D & A_0^\dagger U_{01} \\
            \hline
            0 & U_{10}B_0^\dagger & U_{11}
        \end{array}\right]\cdot\left[\begin{array}{c|c|c}
            \idmat{k} & 0 & 0 \\
            \hline
            0 & B_0 & 0 \\
            \hline
            0 & 0 & \idmat{n}
        \end{array}\right]
    \end{gather*}
    Using respectively the RQD and QLD, we get
    \begin{gather*}
        \left[\begin{array}{c}0 \\ \hline A_0^\dagger U_{01}\end{array}\right]=RQ
        \hspace{3em}\textup{and}\hspace{3em}
        \left[\begin{array}{c|c}0 & U_{10}B_0^\dagger\end{array}\right]=A_1L
    \end{gather*}
    where $R:n\times n$ is upper-triangular, $L:n\times n$ is lower-triangular and $A_1,Q:n\times n$ are unitaries. Then, the middle matrix can be further decomposed as follows.
    \begin{gather*}
        \left[\begin{array}{c|c}
            \begin{array}{c|c}
                \idmat{k} & 0 \\
                \hline
                0 & D
            \end{array} & RQ \\
            \hline
            A_1L & U_{11}
        \end{array}\right]
        =\left[\begin{array}{c|c|c}
            \idmat{k} & 0 & 0 \\
            \hline
            0 & \idmat{n-k} & 0 \\
            \hline
            0 & 0 & A_1
        \end{array}\right]\cdot\left[\begin{array}{c|c}
            \begin{array}{c|c}
                \idmat{k} & 0 \\
                \hline
                0 & D
            \end{array} & R \\
            \hline
            L & A_1^\dagger U_{11}Q^\dagger
        \end{array}\right]\cdot\left[\begin{array}{c|c|c}
            \idmat{k} & 0 & 0 \\
            \hline
            0 & \idmat{n-k} & 0 \\
            \hline
            0 & 0 & Q
        \end{array}\right]
    \end{gather*}
    The unitarity of $\hat{U}$ forces the elements of the diagonal matrix $D$ to be between $0$ and $1$, and the SVD sorts them from largest to smallest. Let $\ell\ge k$ be such that the first $\ell-k$ diagonal elements of $D$ are $1$s and let $D_{00}:(n-\ell)\times(n-\ell)$ be the submatrix of $D$ which has only elements strictly smaller than $1$, which implies
    \begin{gather*}
        \left[\begin{array}{c|c}
            \idmat{k} & 0 \\
            \hline
            0 & D
        \end{array}\right]=\left[\begin{array}{c|c}
            \idmat{\ell} & 0 \\
            \hline
            0 & D_{00}
        \end{array}\right].
    \end{gather*}
    Then, we can split $V\defeq A_1^\dagger U_{11}Q^\dagger:n\times n$ into four blocks $V_{00}:\ell\times\ell$, $V_{01}:\ell\times(n-\ell)$, $V_{10}:(n-\ell)\times\ell$ and $V_{11}:(n-\ell)\times(n-\ell)$. Moreover, the orthonormality of the $\ell$ first rows and columns of $\hat{U}$ forces the middle matrix to have the form
    \begin{gather*}
        \left[\begin{array}{c|c|c|c}
            \idmat{\ell} & 0 & 0 & 0 \\
            \hline
            0 & D_{00} & 0 & D_{01}\\
            \hline
            0 & 0 & V_{00} & V_{01}\\
            \hline
            0 & D_{10} & V_{10} & V_{11}
        \end{array}\right]
    \end{gather*}
    where $D_{01}:(n-\ell)\times(n-\ell)$ is the resulting upper-triangular submatrix of $R$ and $D_{10}:(n-\ell)\times(n-\ell)$ is the resulting lower-triangular submatrix of $L$. The unitarity of $\hat{U}$ implies that 
    \begin{gather*}
        D_{00}D_{00}^\dagger+D_{01}D_{01}^\dagger=\idmat{n-\ell}
        \hspace{3em}\textup{and}\hspace{3em}
        D_{00}^\dagger D_{00}+D_{10}^\dagger D_{10}=\idmat{n-\ell}
    \end{gather*}
    where $D_{00}D_{00}^\dagger=D_{00}^\dagger D_{00}=D_{00}^2$ is diagonal with only non-negative strictly smaller than $1$ elements. Thus, $D_{01}D_{01}^\dagger$ and $D_{10}^\dagger D_{10}$ are diagonal with non-zero diagonal elements and \cref{lem:triangulardiagonal} forces $D_{01}$ and $D_{10}$ to be diagonal (with real non-negative diagonal elements as the diagonal elements of the RQD and QLD are so). Thus, $D_{00}^2+D_{01}^2=\idmat{n-\ell}=D_{00}^2+D_{10}^2$ which implies $D_s'\defeq D_{10}=D_{01}$. Notice that all diagonal elements of $D_s'$ are $>0$. Moreover, the unitarity of $\hat{U}$ implies $V_{01}D_s'=0$ and $D_s'V_{10}=0$, which forces $V_{01}=0$ and $V_{10}=0$.
    Let $D_c'\defeq D_{00}$. Again, the unitarity of $\hat{U}$ implies $D_s'D_c'+V_{11}D_s'=0$, which forces $V_{11}=-D_c'$. Then, the middle matrix can be further decomposed as follows.
    \begin{gather*}
        \left[\begin{array}{c|c|c|c}
            \idmat{\ell} & 0 & 0 & 0 \\
            \hline
            0 & D_c' & 0 & D_s'\\
            \hline
            0 & 0 & V_{00} & 0\\
            \hline
            0 & D_s' & 0 & -D_c'
        \end{array}\right]
        =\left[\begin{array}{c|c|c|c}
            \idmat{\ell} & 0 & 0 & 0 \\
            \hline
            0 & D_c' & 0 & -D_s'\\
            \hline
            0 & 0 & \idmat{\ell} & 0\\
            \hline
            0 & D_s' & 0 & D_c'
        \end{array}\right]\cdot\left[\begin{array}{c|c|c|c}
            \idmat{\ell} & 0 & 0 & 0 \\
            \hline
            0 & \idmat{n-\ell} & 0 & 0\\
            \hline
            0 & 0 & V_{00} & 0\\
            \hline
            0 & 0 & 0 & -\idmat{n-\ell}
        \end{array}\right]
    \end{gather*}
    Finally, we can pad $D_c'$ with $1$s and $D_s'$ with $0$s to get $D_c,D_s:(n-k)\times(n-k)$ such that
    \begin{gather*}
        \left[\begin{array}{c|c}
            \idmat{\ell} & 0 \\
            \hline
            0 & D_c'
        \end{array}\right]=\left[\begin{array}{c|c}
            \idmat{k} & 0 \\
            \hline
            0 & D_c
        \end{array}\right]
        \hspace{3em}\textup{and}\hspace{3em}
        \left[\begin{array}{c|c}
            0_\ell & 0 \\
            \hline
            0 & D_s'
        \end{array}\right]=\left[\begin{array}{c|c}
            0_k & 0 \\
            \hline
            0 & D_s
        \end{array}\right]
    \end{gather*}
    where $0_k$ represents the $k\times k$ zero matrix. All put together we get the following decomposition.
    \begin{gather*}
        \hat{U}=\left[\begin{array}{c|c|c}
            \idmat{k} & 0 & 0 \\
            \hline
            0 & A_0 & 0 \\
            \hline
            0 & 0 & A_1
        \end{array}\right]\cdot\left[\begin{array}{c|c|c|c}
            \idmat{k} & 0 & 0 & 0 \\
            \hline
            0 & D_c & 0 & -D_s\\
            \hline
            0 & 0 & \idmat{k} & 0\\
            \hline
            0 & D_s & 0 & D_c
        \end{array}\right]\cdot\left[\begin{array}{c|c|c}
            \idmat{n} & 0 & 0 \\
            \hline
            0 & V_{00} & 0\\
            \hline
            0 & 0 & -\idmat{n-\ell}
        \end{array}\right]\cdot\left[\begin{array}{c|c|c}
            \idmat{k} & 0 & 0 \\
            \hline
            0 & B_0 & 0 \\
            \hline
            0 & 0 & Q
        \end{array}\right]
    \end{gather*}
    And we are done with $B_1\defeq\left[\begin{array}{c|c}
            V_{00} & 0 \\
            \hline
            0 & -\idmat{n-\ell}
        \end{array}\right]Q$.
\end{proof}

\section{Proof of \cref{th:acqccompleteness}}\label[appendix]{app:acqccompleteness}

In this section we use the notation $\tf{wcontrol}\defeq\tf{wcontroldef}$ to depict a negative control. Notice that \cref{lem:acqcpoints} (together with \labelcref{eq:notnot}) implies that \tf{binit} and \tf{winit} are points of the negative control, meaning
\begin{equation*}
    \catcanc\vdash\tf{binitcc_wc}=\tf{binitc_wc}
    \hspace{2em}\text{and}\hspace{2em}
    \catcanc\vdash\tf{winitcc_wc}=\tf{winitidid_wc}.
\end{equation*}

To prove the completeness of $\catcanc$, we first prove \cref{lem:ninitidentity}, which can be seen as a generalized version of the axioms \labelcref{eq:ic,,eq:icc}.

\begin{lemma}\label{lem:ninitidentity}
    Given a controllable quantum circuit $C\in\catcuni(n+k,n+k)$, if for any $x\in\{0,1\}^k$ we have $\intcuni{C}\ket{0^nx}=\ket{0^nx}$, then 
    \begin{gather*}
        \catcanc\vdash\tf{ninitidentity_left}=\tf{ninitidentity_right}
    \end{gather*}
\end{lemma}
\begin{proof}
    By induction on $n$, where the case $n=0$ is trivial because there is no \tf{winit} and we can directly use the completeness of $\catcuni$ (see \cref{th:cqccompleteness}). If $n=1$, then the condition implies that there exists a controllable quantum circuit $C'\in\catcuni(k,k)$ such that $\intcuni{C}=\intcuni{\ctrl(C')}$; and we can directly conclude using the completeness of $\catcuni$ (see \cref{th:cqccompleteness}) together with \cref{lem:acqcpoints}. The induction case requires a finer matrix decomposition, namely the GCSD (see \cref{lem:gcsd}). Suppose $C\in\catcuni(1+n+k,1+n+k)$.  Let $K\defeq 2^k$ and $N\defeq 2^{n+k}$. The condition implies that there exists a unitary matrix $U:(2N-K)\times(2N-K)$ such that
    \begin{gather*}
        \intiso{C}=\left[\begin{array}{c|c}
            \idmat{K} & 0 \\
            \hline
            0 & U
        \end{array}\right].
    \end{gather*}
    \cref{lem:gcsd} gives the following decomposition
    \begin{gather*}
        \intiso{C}=\left[\begin{array}{c|c|c}
            \idmat{K} & 0 & 0 \\
            \hline
            0 & A_0 & 0 \\
            \hline
            0 & 0 & A_1
        \end{array}\right]\cdot\left[\begin{array}{c|c|c|c}
            \idmat{K} & 0 & 0 & 0 \\
            \hline
            0 & D_c & 0 & -D_s\\
            \hline
            0 & 0 & \idmat{K} & 0\\
            \hline
            0 & D_s & 0 & D_c
        \end{array}\right]\cdot\left[\begin{array}{c|c|c}
            \idmat{K} & 0 & 0 \\
            \hline
            0 & B_0 & 0 \\
            \hline
            0 & 0 & B_1
        \end{array}\right]
    \end{gather*}
    where $A_0,B_0:(N-K)\times (N-K)$ and $A_1,B_1:N\times N$ are unitary matrices, and $D_c,D_s:(N-K)\times (N-K)$ are real diagonal matrices with non-negative elements satisfying $D_c^2+D_s^2=I_{N-K}$. Notice that there exists $\alpha_i\in[0,\nicefrac{\pi}{2}]$ such that
    \begin{gather*}
        D_c=\diag(\cos(\alpha_1),\dots,\cos(\alpha_{N-K}))
        \hspace{1.5em}\text{and}\hspace{1.5em}
        D_s=\diag(\sin(\alpha_1),\dots,\sin(\alpha_{N-K})).
    \end{gather*}
    Let $R_y(\alpha)\in\catcuni(1,1)$ be such that
    \begin{gather*}
        \intcuni{R_y(\alpha)}=\left[\begin{array}{c|c}
            \cos(\alpha) & -\sin(\alpha) \\
            \hline
            \sin(\alpha) & \cos(\alpha)
        \end{array}\right],
    \end{gather*}
    where $\alpha\in\R$. Moreover, given $b\in\{0,1\}^\ell$, let $\ctrl^b(R_y(\alpha))\in\catcuni(1+\ell,1+\ell)$ be a controlled $R_y(\alpha)$ whose target is the first qubit and whose $(1+j)$-th qubit is positively controlled if the $j$-th digit of $b$ is $1$, and negatively controlled if the $j$-th digit of $b$ is $0$. For instance,
    \begin{gather*}
        \ctrl^{011}(R_y(\alpha))=\tf{ryexample}.
    \end{gather*}
    Then, the following equality holds.
    \begin{gather*}
        \textstyle\intiso{\prod_{K\le i<N}\ctrl^{\bin(i)}(R_y(\alpha_{i-K+1}))}=\left[\begin{array}{c|c|c|c}
            \idmat{K} & 0 & 0 & 0 \\
            \hline
            0 & D_c & 0 & -D_s\\
            \hline
            0 & 0 & \idmat{K} & 0\\
            \hline
            0 & D_s & 0 & D_c
        \end{array}\right]
    \end{gather*}
    where $\bin(i)$ is the binary representation of $i$ of size $n+k$ (for instance, $\bin(K)=0^{n-1}10^{k}$ and $\bin(N-1)=1^{n+k}$). Let $C_{A_0},C_{B_0},C_{A_1},C_{B_1}\in\catcuni(n+k,n+k)$ be such that
    \begin{gather*}
        \interp{C_{A_0}}_{\textup{i}}=\left[\begin{array}{c|c}
            \idmat{K} & 0 \\
            \hline
            0 & A_0
        \end{array}\right], \hspace{2em}
        \interp{C_{B_0}}_{\textup{i}}=\left[\begin{array}{c|c}
            \idmat{K} & 0 \\
            \hline
            0 & B_0
        \end{array}\right], \hspace{2em}
        \interp{C_{A_1}}_{\textup{i}}=A_1, \hspace{2em}
        \interp{C_{B_1}}_{\textup{i}}=B_1.
    \end{gather*}
    Then, by completeness of $\catcuni$ (see \cref{th:cqccompleteness}), we have
    \begin{gather*}
        \catcanc\vdash C=\tf{ninitidentity_01}
    \end{gather*}
    where the notation \tf{bwcontrol} may represent either a positive control or a negative control. Then, the following derivation concludes the proof.
    \begin{gather*}
        \catcanc\vdash\tf{ninitidentity_02}
        \eqcref{th:cqccompleteness}\tf{ninitidentity_03} \\[0.5em]
        \eqcref{lem:acqcpoints}\tf{ninitidentity_04} \\[0.5em]
        \eqihref\tf{ninitidentity_05}
        \eqcref{lem:acqcpoints}\tf{ninitidentity_07}
        \eqihref\tf{ninitidentity_08}
    \end{gather*}
    Notice that all the $\ctrl^{\bin(i)}(R_y(\alpha_{i-K+1}))$ can be removed using \cref{lem:acqcpoints} because every bitstring $\bin(i)$ satisfying $K\le i<N$ has at least one $1$ in its first $n$ digits. Indeed, each controlled gate has at least one positive control in front of an initialization, which ensures that \cref{lem:acqcpoints} is applicable.
\end{proof}

\begin{proof}[Proof of \cref{th:acqccompleteness}]
    Let $C_1,C_2\in\catcanc(k,n+k)$ be such that $\intanc{C_1}=\intanc{C_2}$. Applying \cref{lem:pulloutancilla} to $C_i$ gives the circuit ${C_i'\in\cat{CQC}(m_i+n+k,m_i+n+k)}$. Assume w.l.o.g. that $m_1\ge m_2$. Then, the following derivation proves completeness, where the four wires depict in reality $m_1-m_2$, $m_2$, $n$ and $k$ wires respectively.
    \begin{gather*}
        \catcanc\vdash\tf{iqccompleteness_00}
        \eqcref{lem:pulloutancilla}\tf{iqccompleteness_01}
        \eqcref{lem:cqcunitary}\tf{iqccompleteness_02} \\
        \eqcref{lem:ninitidentity}\tf{iqccompleteness_03}
        \eqeqref{eq:a}\tf{iqccompleteness_04}
        \eqcref{lem:pulloutancilla}\tf{iqccompleteness_05}
    \end{gather*}
    where \cref{lem:ninitidentity} is applicable because $\intanc{C_1}=\intanc{C_2}$ implies
    \begin{gather*}
        \intanc{\tf{iqccompleteness_06}}=\intanc{\tf{iqccompleteness_07}}
    \end{gather*}
    which implies
    \begin{gather*}
        \intanc{\tf{iqccompleteness_08}}=\intanc{\tf{iqccompleteness_09}}
    \end{gather*}
    which implies $\intcuni{C_1'\seqop(\id{m_1-m_2}\parop(C_2')^\dagger)}\ket{0^{m_1+n}x}=\ket{0^{m_1+n}x}$ for any $x\in\{0,1\}^k$.
\end{proof}